%% file: SSEC.tex
\documentclass[journal,twoside]{IEEEtran}
\usepackage[T1]{fontenc}% optional T1 font encoding
\usepackage{graphicx, xcolor, amssymb, amsmath, amsfonts, amsthm, cite, balance, mathrsfs}
\usepackage{pgfplots}
\usepackage{pgfplotstable}
\usepackage{tikz}
\usetikzlibrary{pgfplots.groupplots}
\usepackage[font=small]{caption}
\usepackage{subcaption}
\newtheorem{theorem}{Theorem}

\allowdisplaybreaks
\DeclareMathOperator{\argmax}{argmax}
\DeclareMathOperator{\argmin}{argmin}
\begin{document}
\title{\LARGE{On the Effective Rate of NOMA in Underlay Spectrum Sharing}}
\author{Vaibhav~Kumar,~\IEEEmembership{Member,~IEEE,} Zhiguo~Ding,~\IEEEmembership{Fellow,~IEEE,} and~Mark~F.~Flanagan,~\IEEEmembership{Senior Member,~IEEE}% <-this % stops a space
\thanks{Copyright (c) 2015 IEEE. Personal use of this material is permitted. However, permission to use this material for any other purposes must be obtained from the IEEE by sending a request to pubs-permissions@ieee.org. \\ V.~Kumar and M.~F.~Flanagan are with School of Electrical and Electronic Engineering, University College Dublin, Ireland (e-mail: vaibhav.kumar@ieee.org, mark.flanagan@ieee.org).\\
Z.~Ding is with School of Electrical and Electronic Engineering, The University of Manchester, UK (e-mail: zhiguo.ding@manchester.ac.uk).}% <-this % stops a space
%\thanks{X.}% <-this % stops a space
}

% instance, "\textbf{A} \textbf{B}" will typeset as "A B" not "AB". To 
\markboth{IEEE Transactions on Vehicular Technology}%
{Kumar \MakeLowercase{\textit{et al.}}: On the effective rate of NOMA in underlay spectrum sharing}
\maketitle
\begin{abstract}
In this paper, we present the delay-constrained performance analysis of a multi-antenna-assisted multiuser non-orthogonal multiple access (NOMA) based spectrum sharing system over Rayleigh fading channels. We derive analytical expressions for the sum effective rate (ER) for the downlink NOMA system under a peak interference constraint. In particular, we show the effect of the availability of different levels of channel state information (instantaneous and statistical) on the system performance. We also show the effect of different parameters of interest, including the peak tolerable interference power, the delay exponent, the number of antennas and the number of users, on the sum ER of the system under consideration. An excellent agreement between simulation and theoretical results confirms the accuracy of the analysis. 
\end{abstract}

%\begin{IEEEkeywords}
%Communications Society, IEEE, IEEEtran, journal, \LaTeX, paper, template.
%\end{IEEEkeywords}
\IEEEpeerreviewmaketitle

\section{Introduction}
\IEEEPARstart{W}{ith} an explosive growth in the number of wireless communication devices and users, along with an unprecedented growth in internet traffic, one of the major challenges for the development of beyond-fifth-generation (B5G) systems is to accommodate a massive number of devices into a congested spectrum. Moreover, with the development of new wireless services, there is a new wave of applications which are extremely delay sensitive. Such applications include live video streaming, online gaming, vehicle-to-everything (V2X) communications and tactile Internet. NOMA and spectrum sharing are two major technologies that can mitigate the problem of spectrum scarcity while also providing massive connectivity with low latency. It is well-known that NOMA facilitates spectrum sharing and can also enhance the secrecy performance of unmanned aerial vehicle (UAV) aided simultaneous wireless information and power transfer (SWIPT)~\cite{UAV} as well as that of the primary/licensed network~\cite{Alignment}.

The average achievable rate is one of the commonly used metrics to evaluate the performance of a wireless communication system. However, it is important to note that this metric is based on Shannon's channel capacity formula, which does not take into account the queuing delay at the transmitter. In order to quantify the delay-constrained performance of a wireless communications system, the concept of effective rate (also known as effective capacity\footnote{The authors in~\cite{Negi} defined the effective capacity as the maximum constant data arrival rate that can be supported by a radio link while satisfying the delay constraint.} or link-layer capacity) was introduced in~\cite{Negi}. The ER analysis of different OMA-based spectrum sharing systems has been well-investigated in many notable works including~\cite{EC_Nakagami} and~\cite{EC_band_selection_new}. In~\cite{EC_Nakagami}, the authors presented a closed-form analysis of an underlay spectrum sharing system over Nakagami-$m$ fading, considering an average interference constraint to ensure the quality-of-service (QoS) of the primary user. The ER analysis of a spectrum sharing system with an average interference constraint and different spectrum-band selection criteria was presented in~\cite{EC_band_selection_new}.

On the other hand, the delay-constrained performance analysis of NOMA systems has gained significant attention in the past few years. Among many notable contributions, some recent ones are given in~\cite{EC_NOMA, EC_JSTSP, EC_alpha_mu}. In~\cite{EC_NOMA}, the authors presented a closed-form analysis of the ER for a multiuser downlink NOMA system assuming that the \emph{instantaneous} channel state information (CSI) is available at the transmitter. A closed-form analysis of the delay violation probability and ER of a two-user downlink NOMA system over Nakagami-$m$, Rician and $\alpha$-$\mu$ fading channels was presented in~\cite{EC_JSTSP, EC_alpha_mu}, where the authors assumed that only the \emph{statistical} CSI of the downlink channels is available at the transmitter. 
%The ER analysis of a multiuser downlink NOMA system in finite-blocklength regime was presented in~\cite{EC_finite_blocklength}.

The superiority of NOMA in an underlay spectrum sharing scenario over its OMA-based counterpart, in terms of the average achievable rate, outage probability and ergodic sum secrecy rate, has been well-established in~\cite{Kumar_TCOM,Kumar_TVT,ICT_Invited}. However, to the best of the authors' knowledge, the delay-constrained analysis of a spectrum sharing NOMA system has not yet been presented in the open literature. Therefore, motivated by this, in this paper, we present the sum ER analysis of an underlay spectrum sharing multiuser downlink NOMA system, over Rayleigh fading channels, where the secondary users are assumed to be equipped with multiple receive antennas. We consider different levels of CSI availability at the secondary transmitter (ST), as given in~Table~\ref{TableCases}, and derive exact closed-form expressions for all of the cases to quantify the sum ER of the downlink NOMA system. Our results show the effect of different parameters of interests, including the peak tolerable interference power at the primary receiver (PR), the delay exponent, the number of secondary users, and the number of antennas at the secondary receivers, on the system performance. We also compare the performance of the NOMA-based system under consideration with a corresponding orthogonal time-frequency-division multiple access system. Moreover, we use a bisection-search-based optimal power allocation scheme that ensures that the ER of the strong users (to be defined later) in the NOMA-based system is equal to that of the strong users in the corresponding OMA-based system, and then the remaining power is distributed equally among the weak users (also to be defined later) in the NOMA system.
%=====================================
\section{System Model}
We consider an underlay spectrum sharing downlink NOMA system consisting of a ST, a PR, and $K$ secondary users, denoted by $\mathrm{U}_1, \mathrm U_2, \ldots, \mathrm U_K$ (here $K$ is assumed to be an even positive  integer). The secondary network consisting of the ST and users $\mathrm U_k, k \in \{1, 2, \ldots, K\}$, is assumed to be deployed outside the coverage area of the primary network. It is assumed that ST and PR are each equipped with a single antenna, while each user $\mathrm U_k$ is equipped with $N_k(\geq 1)$ antennas. Let $d_p$ and $d_k$ denote the distance between the ST and PR, and that between the ST and $\mathrm U_k$, respectively. Furthermore, assume that $h_{\mathrm p} \triangleq \sqrt{\Omega_{\mathrm p}}\hat h_{\mathrm p}$ and $h_{k,i} \triangleq \sqrt{\Omega_k} \hat h_{k, i}$ denote the channel coefficient between the ST and PR, and between the ST and $i$-th antenna of $\mathrm U_k$ (here $1 \leq i \leq N_k$), respectively. Here $\hat h_{\mathrm p} \sim \mathcal{CN}(0, 1)$, $\hat h_{k, i} \sim \mathcal{CN}(0, 1)$, and $\Omega_{\mathrm p}$ and $\Omega_k$ denote the path loss associated with the ST-PR link and that associated with the link between the ST and the $i$-th antenna at $\mathrm U_k$, respectively. The corresponding channel gains are denoted by $g_{\mathrm p} \triangleq |h_{\mathrm p}|^2$ and $g_{k, i} \triangleq |h_{k, i}|^2$. It is also assumed that the total available bandwidth for the secondary (NOMA) users is denoted by $B$, and the length of each fading block corresponding to links between the ST and the secondary users is denoted by $T$. In this work, we consider that there is no constraint on the power budget at the ST, i.e., the ST, in principle, has unlimited power; however, the power transmitted from the ST is constrained by a peak interference constraint at the PR. We also assume that instantaneous CSI is available at $\mathrm U_k$ regarding the ST-$\mathrm U_k$ links, whereas the details regarding the availability of CSI at the ST for different links in different cases is given in~Table~\ref{TableCases}. Note that in the table, ``IL-CSI'' (interference-link CSI) denotes the CSI for the ST-PR link, whereas ``SL-CSI'' (secondary-link CSI) denotes the CSI for the ST-$\mathrm U_k$ links.

Following the arguments in~\cite{UserPairing}, we adopt a low-overhead location-based user-pairing scheme. In this scheme, we form a set $\mathcal S$, consisting of the distances of the secondary users from the ST, arranged in a descending order. Then the first user-pair consists of the secondary users whose distances are the first and the last elements in $\mathcal S$. After constructing the first user-pair, the distance of the selected users is removed from $\mathcal S$. This process is repeated until $\mathcal S = \varnothing$.

Note that in the sections to follow, we present the performance analysis for an arbitrary user-pair consisting of two users $\mathrm{U_n}$ (the \emph{near} user) and $\mathrm{U_f}$ (the \emph{far} user). Note that in each of the user-pairs, the user that is nearer to ST is labeled as $\mathrm{U_n}$ and the other user is labeled as $\mathrm{U_f}$. We allocate equal bandwidth for all the user-pairs, i.e., the bandwidth available for each user-pair is given by $\mathcal B \triangleq 2B/K$. However, in~Section~\ref{sec:Results}, we show the results for the overall system, by adding the sum ER of each of the user-pairs.
Throughout the text, $f_{\mathcal X}(\cdot)$, $F_{\mathcal X}(\cdot)$ and $F^{-1}_{\mathcal X}(\cdot)$ denote the probability density function (PDF), cumulative distribution function (CDF) and inverse distribution function (IDF), respectively, for the random variable $\mathcal X$.

Next, we will provide a detailed ER analysis for the different system configurations given in~Table~\ref{TableCases}.
\begin{table}[t]
\centering 
\caption{Different NOMA system configurations analyzed in this work.}
\label{TableCases}
\begin{tabular}{|l|l|l|}
\hline
\textbf{Case} & \textbf{IL-CSI at the ST} & \textbf{SL-CSI at the ST} \\ \hline
I-I            & Instantaneous             & Instantaneous             \\ \hline
I-S            & Instantaneous             & Statistical               \\ \hline
S-I            & Statistical               & Instantaneous             \\ \hline
S-S            & Statistical               & Statistical               \\ \hline
\end{tabular}
\end{table}

\section{Case I--I: Instantaneous IL-CSI and SL-CSI} \label{sec:I-I}
In this section, we consider the scenario where the ST has instantaneous CSI regarding the ST-PR link as well as the ST-$\mathrm U_{\mathrm n}$ and ST-$\mathrm U_{\mathrm f}$ links. Due to the availability of instantaneous SL-CSI, the users $\mathrm{U_n}$ and $\mathrm{U_f}$ are further categorized as the \emph{strong} user $\mathrm{U_s}$ and \emph{weak} user $\mathrm{U_w}$, where $\mathrm s \triangleq \argmax_{j\in\{\mathrm n, \mathrm f\}} \sum_{i = 1}^{N_j}|h_{j, i}|^2$, and $\mathrm w \triangleq \argmin_{j\in\{\mathrm n, \mathrm f\}} \sum_{i = 1}^{N_j}|h_{j, i}|^2$. 

The ST transmits a power-division multiplexed symbol $\sqrt{a_{\mathrm s} P_{\mathrm t}(g_{\mathrm p})} x_{\mathrm s} + \sqrt{a_{\mathrm w} P_{\mathrm t}(g_{\mathrm p})} x_{\mathrm w}$ to both the users. Here $P_{\mathrm t}(g_{\mathrm p})$, which is in general a one-to-one mapping from $g_{\mathrm p}$ to the set of positive real numbers~$\mathbb R^+$, denotes the power transmitted from the ST, $x_{\mathrm s}$ and $x_{\mathrm w}$ denote the symbol intended for the strong and the weak user, respectively, and $a_{\mathrm s}$ and $a_{\mathrm w}$ are the corresponding power allocation coefficients with $a_{\mathrm s} < a_{\mathrm w}$ and $a_{\mathrm s} + a_{\mathrm w} = 2/K$. Denoting the peak tolerable interference power at the PR by $I$, the optimal instantaneous transmit power to maximize the sum ER of the secondary network is given by $P_{\mathrm t}^*(g_{\mathrm p}) = I/g_{\mathrm p}$. Upon reception of the superimposed signal from the ST, the secondary users apply maximal-ratio combining (MRC) to combine the received signal, where the equivalent channel gain between the ST and $\mathrm{U_s}$, and that between the ST and $\mathrm{U_w}$, are respectively given by $g_{\mathrm s} = \sum_{i = 1}^{N_{\mathrm s}} |h_{\mathrm s, i}|^2$ and $g_{\mathrm w} = \sum_{i = 1}^{N_{\mathrm w}} |h_{\mathrm w, i}|^2$. The strong user first decodes the message intended for the weak user (by considering the inter-user interference as additional Gaussian noise) and then applies successive interference cancellation (SIC) to decode its own signal. On the other hand, the weak user directly decodes its own message by considering the inter-user interference as additional noise. Denoting by $\gamma_{\mathrm s}$ the signal-to-noise ratio (SNR) for the decoding of the intended signal at the strong user, and by $\gamma_{\mathrm w}$ the signal-to-interference-plus-noise ratio (SINR) for the corresponding decoding at the weak user, we have $\gamma_s \triangleq \!\tfrac{a_{\mathrm s} I g_{\mathrm s}/g_{\mathrm p}}{\mathscr N_0 \mathcal B} = a_{\mathrm s} \hat I X_{\mathrm s}$ and $\gamma_{\mathrm w} \triangleq \!\tfrac{a_{\mathrm w} I g_{\mathrm w}/g_{\mathrm p}}{\mathscr N_0 \mathcal B + a_{\mathrm s} I g_{\mathrm w}/g_{\mathrm p}} \!=\! \tfrac{a_{\mathrm w} \hat I X_{\mathrm w}}{1 + a_{\mathrm s} \hat I X_{\mathrm w}}$, where $\mathscr N_0$ is the noise power spectral density (PSD), $\hat I \triangleq I/(\mathscr N_0 \mathcal B)$, $X_{\mathrm s} \triangleq g_{\mathrm s}/g_{\mathrm p}$, and $X_{\mathrm w} \triangleq g_{\mathrm w}/g_{\mathrm p}$. Therefore, the ER for $\mathrm U_u, u \in \{\mathrm s, \mathrm w\}$, is given by~\cite{EC_NOMA}
\begin{align}
	R_u = \dfrac{-1}{\nu} \log_2 \left[ \mathbb E_{\gamma_u} \left\{ (1 + \gamma_u)^{-\nu}\right\}\right], \label{Eqn:Ru_def}
\end{align}
where $\nu \triangleq \theta T \mathcal B/\ln 2$ and $\theta$ is the delay QoS exponent, which represents the asymptotic decay-rate of the buffer occupancy at the ST. 

Using~\eqref{Eqn:Ru_def}, the ER for the strong user $\mathrm{U_s}$ is given by 
\begin{align}
	R_{\mathrm s} = \dfrac{-1}{\nu} \log_2 \left[ \int_0^\infty (1 + a_{\mathrm s} \hat I x)^{-\nu} f_{X_{\mathrm s}}(x) \mathrm dx \right]. \label{Eqn:CaseII:Rs_def}
\end{align}
\begin{theorem}\label{CaseII:Theorem_Rs}
	When instantaneous IL-CSI and SL-CSI are available at the ST, the ER of the strong user is given by 
	\begin{align}
		R_{\mathrm s} = -\dfrac{1}{\nu} & \log_2 \left[ T_1 (N_{\mathrm n}, \Omega_{\mathrm n}) + T_1 (N_{\mathrm f}, \Omega_{\mathrm f}) \right. \notag \\
		& \left. - T_2(N_{\mathrm n}, N_{\mathrm f}, \Omega_{\mathrm n}, \Omega_{\mathrm f}) - T_2(N_{\mathrm f}, N_{\mathrm n}, \Omega_{\mathrm f}, \Omega_{\mathrm n})\right], \label{Eqn:CaseII:RsClosed}
	\end{align}
where $T_1(\alpha, \beta) = \left[ \Omega_{\mathrm p}/(\beta a_{\mathrm s} \hat I)\right]^\alpha  G_{2, 2}^{2, 2} \left( \tfrac{\Omega_{\mathrm p}}{\beta a_{\mathrm s} \hat I} \left\vert \begin{smallmatrix} -\alpha, & 1 - \alpha \\ 0, & \nu - \alpha \end{smallmatrix} \right. \right) / [ \Gamma(\alpha)$ $ \Gamma(\nu) ]$, $T_2(\alpha_1, \alpha_2, \beta_1, \beta_2) = \sum_{\tau = 0}^{\alpha_2 - 1} \tfrac{[\Omega_{\mathrm p}/(a_{\mathrm s} \hat I)]^{\alpha_1 + \tau}}{\Gamma(\alpha_1) \Gamma(\nu) \beta_1^{\alpha_1} \beta_2^{\tau} \tau!}$ $G_{2, 2}^{2, 2} \left( \tfrac{\Omega_{\mathrm p}}{\Omega a_{\mathrm s} \hat I} \left\vert \begin{smallmatrix} -\alpha_1 - \tau, & 1 - \alpha_1 - \tau \\ 0, & \nu - \alpha_1 - \tau\end{smallmatrix}\right. \right)$, $\Omega \triangleq \Omega_1 \Omega_2 / (\Omega_1 + \Omega_2)$, and $\Gamma[\cdot]$ and $G(\cdot)$ denote the Gamma function and Meijer's G-function, respectively. 
\end{theorem}
\begin{proof}
	See Appendix~\ref{CaseII:Proof_Rs}.
\end{proof}
On the other hand, the ER for the weak user is given by 
\begin{small}
\begin{align}
	& R_{\mathrm w} = \dfrac{-1}{\nu} \log_2 \left[ \int_0^\infty \left(1 + \dfrac{a_{\mathrm w} \hat I x}{1 + a_{\mathrm s} \hat I x}\right)^{-\nu} f_{X_{\mathrm w}}(x) \mathrm dx\right] \notag \\
	= & \ \dfrac{-1}{\nu} \log_2 \left[ \int_0^\infty \left(1 + \tfrac{2}{K}\hat Ix\right)^{-\nu}\left(1 + a_{\mathrm s} \hat I x\right)^{\nu} f_{X_{\mathrm w}}(x) \mathrm dx\right]. \label{Eqn:CaseII:Rw_def}
\end{align}
\end{small}
\begin{theorem} \label{CaseII:Theorem_Rw}
	When instantaneous IL-CSI and SL-CSI are available at the ST, the ER of the weak user is given by
	\begin{align}
		\!\!\!R_{\mathrm w} \!=\!  \dfrac{-1}{\nu} \log_2[T_3(N_{\mathrm n}, N_{\mathrm f}, \Omega_{\mathrm n}, \Omega_{\mathrm f}) + T_3(N_{\mathrm f}, N_{\mathrm n}, \Omega_{\mathrm f}, \Omega_{\mathrm n})],\! \label{Eqn:CaseII:RwClosed}
	\end{align}
where $T_3(\alpha_1, \alpha_2, \beta_1, \beta_2) = \sum_{\tau = 0}^{\alpha_2 - 1} \tfrac{\Omega_{\mathrm p} (2\hat I/K)^{-(\alpha_1 + \tau)}}{\Gamma(\alpha_1) \Gamma(\nu) \Gamma(-\nu) \beta_1^{\alpha_1} \beta_2^{\tau} \tau!}$ $\mathcal G_{1, 1:1, 1:1, 1}^{1, 1:1, 1:1, 1} \left(\begin{smallmatrix} 1 - (\alpha_1 + \tau) \\ \nu - (\alpha_1 + \tau)\end{smallmatrix} \left\vert \begin{smallmatrix} 1 - \nu \\ 0 \end{smallmatrix} \right\vert \left. \begin{smallmatrix} -(\alpha_1 + \tau) \\ 0 \end{smallmatrix} \right\vert \tfrac{K}{2}a_{\mathrm s}, \tfrac{K\Omega_{\mathrm p}}{2\Omega \hat I}  \right),$
and 
$\mathcal G(\cdot)$ denotes the extended generalized bivariate Meijer G-function (EGBMGF).
\end{theorem}
\begin{proof}
	See Appendix~\ref{CaseII:Proof_Rw}.
\end{proof}
\noindent The sum ER, $R_{\mathrm{sum}} = R_{\mathrm s} + R_{\mathrm w}$, for a user-pair of the downlink NOMA system is then obtained by adding~\eqref{Eqn:CaseII:RsClosed} and~\eqref{Eqn:CaseII:RwClosed}.
%=======================================
\section{Case I--S: Instantaneous IL-CSI and Statistical SL-CSI}  \label{sec:I-S}
In this section, we consider the scenario where only statistical CSI\footnote{Statistical CSI refers to knowledge of the channel distribution and the average channel gain.} is available at the ST regarding the ST-$\mathrm{U_n}$ and ST-$\mathrm{U_f}$ links, whereas instantaneous CSI is available for the ST-PR link. In this case, $\mathrm{U_n}$ is considered as the strong user $\mathrm{U_s}$, and $\mathrm{U_f}$ is considered as the weak user $\mathrm{U_w}$. Similar to the previous case, the instantaneous power transmitted from the ST is given by $P_{\mathrm t}^*(g_{\mathrm p}) = I/g_{\mathrm p}$. Following the NOMA principle, the SNR at $\mathrm{U_s}$ to decode $x_{\mathrm s}$ is given by $\gamma_{\mathrm s} = a_{\mathrm s} \hat I g_{\mathrm n}/g_{\mathrm p}$. On the other hand, since $x_\mathrm w$ needs to be decoded at both $\mathrm{U_s}$ and $\mathrm{U_w}$, the SINR to decode $x_\mathrm{w}$ is given by $\gamma_{\mathrm w} = (a_{\mathrm w} \hat I g_{\mathrm w}/g_{\mathrm p})/(1 + a_{\mathrm s} \hat I g_{\mathrm w}/g_{\mathrm p}) = a_{\mathrm w} \hat I X_{\mathrm w}/(1 + a_{\mathrm s} \hat I X_{\mathrm w})$.

Using the expressions for the PDFs of $g_{\mathrm p}$ and $g_{u'}, u' \in \{\mathrm n, \mathrm f\}$ (given in Appendix~\ref{CaseII:Proof_Rs}), the PDF of the random variable $X_{\mathrm n} \triangleq g_{\mathrm n}/g_{\mathrm p}$ is given by $f_{X_{\mathrm n}}(x) = \int_0^\infty  y f_{g_{\mathrm n}}(xy) f_{g_{\mathrm p}}(y) \mathrm dy 	=  \tfrac{N_{\mathrm n} x^{N_{\mathrm n} - 1}}{\Omega_{\mathrm n}^{N_{\mathrm n}} \Omega_{\mathrm p}} \left( \tfrac{x}{\Omega_{\mathrm n}} + \tfrac{1}{\Omega_{\mathrm p}} \right)^{-(N_{\mathrm n} + 1)}$.
Using the expression for $f_{X_{\mathrm n}}(x)$ and solving the integral by following the steps given in~Appendix~\ref{CaseII:Proof_Rs}, the ER for the strong user is given by 
\begin{small}
\begin{align}
	& R_{\mathrm s} = \dfrac{-1}{\nu} \log_2 \left[ \int_0^\infty (1 + a_{\mathrm s} \hat I x)^{-\nu} f_{X_{\mathrm n}} (x) \mathrm dx\right] \notag \\
	= & \ \dfrac{-1}{\nu} \log_2 \left[ \dfrac{[\Omega_{\mathrm p}/(\Omega_{\mathrm n} a_{\mathrm s} \hat I)]^{N_{\mathrm n}}}{\Gamma(\nu) \Gamma(N_{\mathrm n})} G_{2, 2}^{2, 2} \left( \dfrac{\Omega_{\mathrm p}}{\Omega_{\mathrm n} a_{\mathrm s} \hat I}\left\vert \begin{smallmatrix} -N_{\mathrm n}, & 1 - N_{\mathrm n} \\[0.6em]0, & \nu - N_{\mathrm n}\end{smallmatrix}\right. \right)\right]. \label{Eqn:CaseIS:RsClosed}
\end{align}
\end{small}

On the other hand, the ER of the weak user will be the same as for the previous case, and its closed-form expression is given by~\eqref{Eqn:CaseII:RwClosed}. However, in the previous case, the weak user is chosen dynamically, whereas for the current case the weak user $\mathrm{U_w}$ is fixed as $\mathrm{U_f}$. The sum ER, $R_{\mathrm{sum}} = R_{\mathrm s} + R_{\mathrm w}$, is then obtained by adding~\eqref{Eqn:CaseIS:RsClosed} and~\eqref{Eqn:CaseII:RwClosed}.
%=======================================
\section{Case S--I: Statistical IL-CSI and Instantaneous SL-CSI}  \label{sec:S-I}
In a practical system, it is sometimes difficult for the secondary network to have a proper coordination with the primary network and therefore, the ST may not have access to instantaneous CSI regarding the ST-PR link. Therefore, in this section, we consider the scenario where the ST has only the statistical CSI, i.e., the knowledge of $\Omega_{\mathrm p}$ as well as the knowledge of the distribution of $g_{\mathrm p}$, along with instantaneous CSI for the ST-$\mathrm{U_n}$ and ST-$\mathrm{U_f}$ links. Therefore, $\mathrm{U_s}$ and $\mathrm{U_w}$ are defined in the same fashion as given in~Section~\ref{sec:I-I}. However, in this case, the QoS at the PR is protected through a statistical constraint.
Denoting the power transmitted from the ST\footnote{Note that in this case, since the ST does not have the instantaneous information regarding $g_{\mathrm p}$, the transmit power from the ST is not a function of $g_{\mathrm p}$.} by $P_{\mathrm t}$, we have $\Pr(g_{\mathrm p} P_{\mathrm t} > I) \leq \delta$ which implies $P_{\mathrm t} \leq I/F_{g_{\mathrm p}}^{-1}(1- \delta)$. Using the fact that $g_{\mathrm p}$ is exponentially distributed with $\mathbb E\{g_{\mathrm p}\} = \Omega_{\mathrm p}$, the maximum allowable transmit power for the ST is given by $P_{\mathrm t}^* = -I/(\Omega_{\mathrm p} \ln \delta)$. Also, following the steps similar to those in~Section~\ref{sec:I-I}, we have $\gamma_{\mathrm s} = \tfrac{a_{\mathrm s}P_{\mathrm t}^* g_{\mathrm s}}{\mathscr N_0 \mathcal B} = a_{\mathrm s}\hat P_{\mathrm t}^* g_{\mathrm s}$ and $\gamma_{\mathrm w} = \tfrac{a_{\mathrm w}P_{\mathrm t}^* g_{\mathrm w}}{\mathscr N_0 \mathcal B + a_{\mathrm s}P_{\mathrm t}^* g_{\mathrm w}} = \tfrac{a_{\mathrm w} \hat P_{\mathrm t}^* g_{\mathrm w}}{1 + a_{\mathrm s} \hat P_{\mathrm t}^* g_{\mathrm w}}$ where $\hat P_{\mathrm t}^* \triangleq P_{\mathrm t}^*/(\mathscr N_0 \mathcal B)$.

The ER for the strong user $\mathrm{U_s}$ is given by 
\begin{align}
	R_{\mathrm s} = \dfrac{-1}{\nu} \log_2 \left[ \int_0^\infty (1 + a_{\mathrm s}  \hat P_{\mathrm t}^* x)^{-\nu} f_{g_{\mathrm s}}(x) \mathrm dx\right]. \label{Eqn:CaseSI:Rs_def}
\end{align}
Substituting the expression for $f_{g_{\mathrm s}}(x)$ (given in~Appendix~\ref{CaseII:Proof_Rs}) into the preceding equation and solving the integral similar to that in~Appendix~\ref{CaseII:Proof_Rs} yields
\begin{align}
		R_{\mathrm s} = -\dfrac{1}{\nu} & \log_2 \left[ T_4 (N_{\mathrm n}, \Omega_{\mathrm n}) + T_4 (N_{\mathrm f}, \Omega_{\mathrm f}) \right. \notag \\
		& \!\!\!\!\!\left. - T_5(N_{\mathrm n}, N_{\mathrm f}, \Omega_{\mathrm n}, \Omega_{\mathrm f}) - T_5(N_{\mathrm f}, N_{\mathrm n}, \Omega_{\mathrm f}, \Omega_{\mathrm n})\right], \label{Eqn:CaseSI:RsClosed}
	\end{align}
where $T_4(\alpha, \beta) = G_{2, 1}^{1, 2} \left( a_{\mathrm s} \hat P_{\mathrm t}^* \beta \left\vert \begin{smallmatrix} 1-\nu, \ 1-\alpha \\ 0 \end{smallmatrix}\right. \right)/[\Gamma(\alpha) \Gamma(\nu)]$ and $T_5(\alpha_1, \alpha_2, \beta_1, \beta_2) = \! \sum_{\tau = 0}^{\alpha_2 - 1} \! \tfrac{\Omega^{\alpha_1 + \tau} G_{2, 1}^{1, 2} \left(\! a_{\mathrm s} \hat P_{\mathrm t}^* \Omega \left\vert \begin{smallmatrix} 1-\nu, \ 1- (\alpha_1 + \tau) \\ 0 \end{smallmatrix} \right. \!\right)}{\Gamma(\alpha_1) \Gamma(\nu) \beta_1^{\alpha_1} \beta_2^{\tau} \tau!}$.
On the other hand, the ER of the weak user is given by 
\begin{align}
	\!\!R_{\mathrm w} \!= \!\dfrac{-1}{\nu} \log_2 \left[ \int_0^\infty \!\!\! (1 + \tfrac{2}{K} \hat P_{\mathrm t}^* x)^{-\nu} (1 + a_{\mathrm s} \hat P_{\mathrm t}^* x)^{\nu} f_{g_{\mathrm w}}(x) \mathrm dx\right]. \label{Eqn:CaseSI:Rw_def}
\end{align}
Substituting the expression for $f_{g_{\mathrm w}}(x)$ (derived in~Appendix~\ref{CaseII:Proof_Rw}) into~\eqref{Eqn:CaseSI:Rw_def} and solving the integral similar to that in~Appendix~\ref{CaseII:Proof_Rw}, we obtain
\begin{align}
	\!\!R_{\mathrm w} \!=\! \dfrac{-1}{\nu}\! \log_2 \left[ T_6(N_{\mathrm n}, \!N_{\mathrm f}, \!\Omega_{\mathrm n}, \!\Omega_{\mathrm f}) \!+\! T_6(N_{\mathrm f},\! N_{\mathrm n},\! \Omega_{\mathrm f},\! \Omega_{\mathrm n})\right],\label{Eqn:CaseSI:RwClosed}
\end{align}
where $T_6(\alpha_1, \alpha_2, \beta_1, \beta_2) = \sum_{\tau = 0}^{\alpha_2 - 1} \tfrac{(2 \hat P_{\mathrm t}^*/K)^{-(\alpha_1 + \tau)}}{\Gamma(\alpha_1) \Gamma(\nu) \Gamma(-\nu) \beta_1^{\alpha_1} \beta_2^{\tau} \tau!}$ $\mathcal G_{1, 1:1, 1:0, 1}^{1, 1:1, 1:1, 0} \left( \begin{smallmatrix} 1- (\alpha_1 + \tau) \\ \nu - (\alpha_1 + \tau) \end{smallmatrix} \left\vert \begin{smallmatrix} 1 + \nu \\ 0 \end{smallmatrix} \right\vert \left. \begin{smallmatrix} - \\ 0 \end{smallmatrix}\right\vert \tfrac{K}{2} a_{\mathrm s}, \tfrac{K}{2\Omega  \hat P_{\mathrm t}^*}\right)$.
\noindent  The sum ER, $R_{\mathrm{sum}} = R_{\mathrm s} + R_{\mathrm w}$, is then obtained by combining~\eqref{Eqn:CaseSI:RsClosed} and~\eqref{Eqn:CaseSI:RwClosed}. 
%=======================================
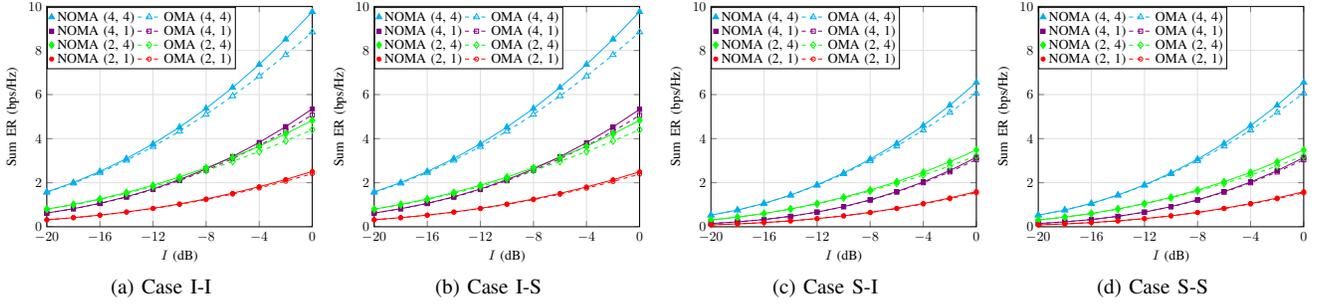
\begin{figure*}[t]
\centering
\begin{subfigure}{.24\textwidth}
  \centering
  \input{Fig1a.tex}
  \vskip-0.2in
  \caption{Case I-I}
  \label{fig:1a}
\end{subfigure}%
\begin{subfigure}{.24\textwidth}
  \centering
  \input{Fig1b.tex}
  \vskip-0.2in
  \caption{Case I-S}
    \label{fig:1b}
\end{subfigure}
\begin{subfigure}{.24\textwidth}
  \centering
  \input{Fig1c.tex}
  \vskip-0.2in
  \caption{Case S-I}
    \label{fig:1c}
\end{subfigure}%
\begin{subfigure}{.24\textwidth}
  \centering
  \input{Fig1d.tex}
  \vskip-0.2in
  \caption{Case S-S}
    \label{fig:1d}
\end{subfigure}
\caption{A comparison of the variation in the sum ER w.r.t. the peak interference power $I$ for the NOMA and OMA systems at $\theta = 1$.}
  \label{fig:1}
\end{figure*}
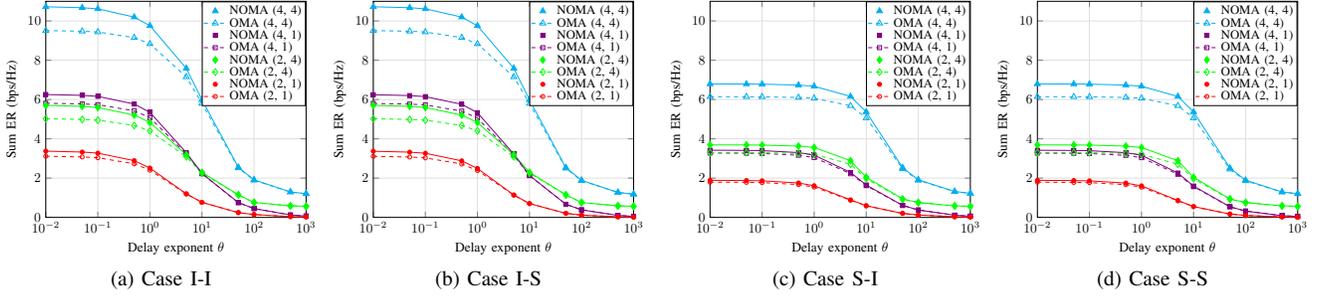
\begin{figure*}
\centering
\begin{subfigure}{.24\textwidth}
  \centering
  \input{Fig2a.tex}
  \vskip-0.2in
  \caption{Case I-I}
  \label{fig:2a}
\end{subfigure}%
\begin{subfigure}{.24\textwidth}
  \centering
  \input{Fig2b.tex}
  \vskip-0.2in
  \caption{Case I-S}
  \label{fig:2b}
\end{subfigure}
\begin{subfigure}{.24\textwidth}
  \centering
  \input{Fig2c.tex}
  \vskip-0.2in
  \caption{Case S-I}
  \label{fig:2c}
\end{subfigure}%
\begin{subfigure}{.24\textwidth}
  \centering
  \input{Fig2d.tex}
  \vskip-0.2in
  \caption{Case S-S}
  \label{fig:2d}
\end{subfigure}
\caption{A comparison of the variation in the sum ER w.r.t. the delay QoS exponent $\theta$ for the NOMA and OMA systems at $I = 0$ dB.}
\label{fig:2}
\end{figure*}
\section{Case S--S: Statistical IL-CSI and Statistical SL-CSI}  \label{sec:S-S}
In this case, we consider the scenario where the ST has statistical CSI for the ST-PR link as well as for the ST-$\mathrm{U_n}$ and ST-$\mathrm{U_f}$ links. Therefore, similar to the case in~Section~\ref{sec:I-S}, the near user $\mathrm{U_n}$ is considered as the strong user $\mathrm{U_s},$ whereas the far user $\mathrm{U_f}$ is considered as the weak user $\mathrm{U_w}$. However, similar to the case in~Section~\ref{sec:S-I}, the power transmitted from the ST is given by $P_{\mathrm t}^* = -I/(\Omega_{\mathrm p} \ln \delta)$. Moreover, we have $\gamma_{\mathrm s} = a_{\mathrm s}\hat P_{\mathrm t}^* g_{\mathrm n}$ and $\gamma_{\mathrm w} = a_{\mathrm w} \hat P_{\mathrm t}^* g_{\mathrm w}/(1 + a_{\mathrm s} \hat P_{\mathrm t}^* g_{\mathrm w})$. The ER for the strong user is given by
\begin{align}
	& R_{\mathrm s} = \dfrac{-1}{\nu} \log_2 \left[ \int_0^\infty (1 + a_{\mathrm s}  \hat P_{\mathrm t}^* x)^{-\nu} f_{g_{\mathrm n}}(x) \mathrm dx \right] \notag \\
	% = & \ \dfrac{-1}{\nu} \log_2 \left[ \int_0^\infty (1 + a_{\mathrm s}  \hat P_{\mathrm t}^* x)^{-\nu} \dfrac{x^{N_{\mathrm n} - 1}}{\Gamma(N_{\mathrm n}) \Omega_{\mathrm n}^{N_{\mathrm n}}} \exp \left( \dfrac{-x}{\Omega_{\mathrm n}}\right) \mathrm dx \right] \notag \\
	= & \ \dfrac{-1}{\nu} \log_2 \left[ \dfrac{1}{\Gamma(N_{\mathrm n}) \Gamma(\nu)} G_{2, 1}^{1, 2} \left( a_{\mathrm s}  \hat P_{\mathrm t}^* \Omega_{\mathrm n} \left\vert \begin{smallmatrix} 1-\nu, 1 - N_{\mathrm n} \\[0.6em] 0\end{smallmatrix} \right. \right) \right]. \label{Eqn:CaseSS:RsClosed}
\end{align}
\noindent The integral above is solved using the expression for $f_{g_{u'}}(x)$ and~\cite[eqns.~(7),~(10),~(11),~(21),~and~(22)]{Reduce}. On the other hand, the analytical expression for the ER of the weak user will be the same as given by~\eqref{Eqn:CaseSI:RwClosed}. 
%However, in the previous section, the weak user is chosen dynamically depending on the instantaneous SL-CSI, whereas, in this section, the near user is considered as the weak user. 
The sum ER, $R_{\mathrm{sum}} = R_{\mathrm s} + R_{\mathrm w}$, is then obtained by combining~\eqref{Eqn:CaseSS:RsClosed} and~\eqref{Eqn:CaseSI:RwClosed}.
%===========RESULTS=====================
\section{Results and Discussion}\label{sec:Results}
In this section, we present numerical and analytical results to investigate the impact of different parameters of interest on the sum ER performance of the underlay NOMA system. The path loss for each of the radio links is modeled as $\mathrm{PL_{ref}} - 10 \xi \log_{10}(d/d_{\mathrm{ref}})$ dB, where $\mathrm{PL_{ref}}$ is the path loss at the reference distance $d_{\mathrm{ref}}$, $d$ denotes the distance between the transmitter and the receiver, and $\xi$ is the path loss exponent. Here we consider $\mathrm{PL_{ref}} = -30$ dB, $d_{\mathrm{ref}} = 1$ m, and the path loss exponent for all of the links is set to $\xi = 2.5$. The PR is assumed to be located at a distance of 40~m from the ST, whereas for the case of secondary users we consider two different scenarios, i.e., $K = 2$ and $K = 4$. For the two-user case, it is assumed that the users are located at 10~m and 50~m from the ST. On the other hand, for the four-user case, the users are considered to be located at 10~m, 12~m, 50~m and 60~m from the ST, respectively. For all cases, we assume $N_k = N$ for all $k \in \{1, 2, \ldots, K\}$; however, the derived analytical expressions hold good for the general case where the secondary users have different numbers of receive antennas. Moreover, we consider the normalized values $B = 1$ Hz, $T = 1$s, $\mathscr N_0 = 1$ W/Hz, and for the case where only statistical SL-CSI is considered to be available at ST, we use $\delta = 0.1$. Note that in all figures, the legends (shown in the figure inset) represent the numerically evaluated results, whereas the solid lines represent the results plotted using the derived analytical expressions for the NOMA systems. Also, in all of the figure legends, the numbers shown in the parentheses denote $(K, N)$.

For the purpose of comparison, we consider a time-frequency-division orthogonal multiple access (OMA) system where the user-pairing is performed in exactly the same fashion as for the NOMA system. Also, for a fair comparison, we consider the same peak power constraint for the OMA system as for the NOMA system. Moreover, we assume equal bandwidth and power allocation among all of the user-pairs. Therefore, in the first time slot, the OMA system serves the strong user in each of the user-pairs simultaneously using orthogonal frequencies, and then in the second time slot, all of the weak users are served simultaneously, but on different frequencies.

It is well-known that for the case of NOMA, the weak/far users normally present a performance bottleneck due to the poor channel conditions and the inter-user interference at such users. In order to overcome this issue, we use a bisection-search-based power allocation scheme for the case of NOMA such that the ER of the strong user in each user-pair is equal to the ER of the strong user in the corresponding OMA system, and then the remaining power is allocated equally among the weak users of each user-pair. This scheme guarantees that the ERs of the strong users in both NOMA and OMA systems are the same, while (as it turns out) the ER of the weak users in the NOMA system is larger than (or approximately equal to) that of the corresponding OMA system, resulting in a higher sum ER in the NOMA system.

Fig.~\ref{fig:1} shows a comparison of the sum ER for NOMA and OMA systems for different cases of CSI availability at the ST. It is evident from the figure that in the low-$I$ regime, the sum ER of both NOMA and OMA systems are almost equal. However, as the value of the peak tolerable interference power (i.e., $I$) increases, the NOMA-based system results in a higher sum ER compared to its OMA-based counterpart. It is also clear from the figure that increasing the number of antennas at the users and/or increasing the number of users results in a larger difference between the sum ER of the NOMA and OMA systems for mid-to-large values of $I$.

In Fig.~\ref{fig:2}, we show the impact of the delay exponent $\theta$ on the sum ER of the NOMA-based and corresponding OMA-based systems for different levels of CSI at the ST. It is interesting to note from the figure that as the value of $\theta$ increases, the sum ER for both NOMA and OMA systems decreases significantly. Note that a large value of $\theta$ represents a highly delay-sensitive system, and for such systems, both the NOMA and OMA systems show similar performance. However, the benefit of a large number of antennas at the secondary users is clearly evident from the figure. On the other hand, $\theta \to 0$ corresponds to a system without any delay constraints, and in this case, the ER becomes equal to the average achievable rate. The figure demonstrates that for low-to-mid values of $\theta$, the NOMA-based system achieves a significantly higher sum ER compared to its OMA-based counterpart.

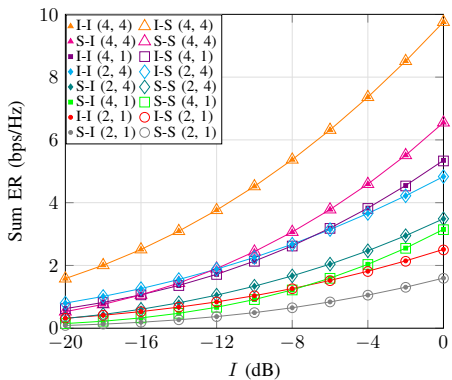
\begin{figure}
\centering
\input{Fig3.tex}
\caption{Comparison of the sum ER for the NOMA system with different levels of CSI at $\theta = 1$.}
\label{fig:3}
\end{figure}
\begin{figure}
\centering
\input{Fig4.tex}
\caption{Comparison of the sum ER for the NOMA system with different levels of CSI at $I = 0$ dB.}
\label{fig:4}
\end{figure}
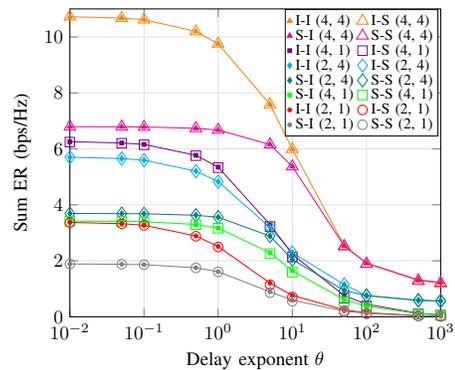

In Fig.~\ref{fig:3}, we compare the effect of different levels of CSI on the sum ER of the NOMA system. It is interesting to note from the figure that the performance of the I-I NOMA system is the same as that of the I-S NOMA system. This means that the NOMA system does not cause any performance degradation when the level of SL-CSI availability reduces from instantaneous to statistical. This occurs due to the fact that when there is a significant difference between the average link quality of the ST-$\mathrm{U_n}$ and ST-$\mathrm {U_f}$ links, the near user and the far user act as the strong user and the weak user, respectively, most of the time, even for the case when instantaneous SL-CSI is available at the ST. A similar trend can be observed while comparing the performance of the S-I NOMA and S-S NOMA systems. Nevertheless, it is also important to note that the system performance degrades significantly when the level of IL-CSI reduces from instantaneous to statistical. Note that for the case of instantaneous IL-CSI availability at the ST, it can transmit at considerably higher power without violating the PR's interference constraint when the ST-PR link is in a deep fading state. On the other hand, for the case when only the statistical IL-CSI is available at the ST, the ST transmits at a constant power, regardless of the instantaneous state of the ST-PR link; this results in a notably inferior performance.

We compare the effect of different levels of CSI at the ST on the variation of sum ER w.r.t. the delay QoS exponent $\theta$ for the NOMA system in~Fig.~\ref{fig:4}. It can be observed from the figure that due to similar reasons as those explained for~Fig.~\ref{fig:3}, the I-I NOMA system performs similar to the I-S NOMA system, and similarly, the performance of the S-I NOMA and S-S NOMA systems are the same, especially for low-to-mid-range values of $\theta$. However, for the case when the delay requirements of the system becomes extremely stringent, i.e., for very large values of $\theta$, all of the considered systems, i.e., I-I NOMA, I-S NOMA, S-I NOMA and S-S NOMA, result in similar sum ER performance. This result indicates that for such extremely delay-intolerant applications, a low-complexity S-S NOMA system provides the same level of performance as its more sophisticated counterparts.
%=======================================
\section{Conclusion}
In this paper, a comprehensive delay-constrained performance analysis of a multi-antenna-assisted multiuser NOMA-based underlay spectrum sharing system was provided considering a peak interference constraint. In particular, we adopted a low-complexity user-pairing scheme, and derived closed-form expressions for the sum ER for all of the cases of different levels of CSI availability at the ST. Our results confirm that for a very small value of the peak tolerable interference power and/or for a large value of the delay QoS exponent, the performance of the NOMA-based system is the same as that of its OMA-based counterpart. However, for a mid-to-large peak interference power and/or small-to-mid value of the delay exponent, the NOMA-based system yields significance performance improvement. On the other hand, our results also establish that when the link qualities of each of the paired users are significantly different, a decreased level of secondary-link CSI at the ST does not have any deleterious effect on the system performance. However, the system performance degrades notably for the case when only the statistical interference-link CSI is available at the ST compared to the case of instantaneous interference-link CSI availability.

%=======================================
\appendices
\section{Proof of Theorem~\ref{CaseII:Theorem_Rs}} \label{CaseII:Proof_Rs}
Given that all links are Rayleigh distributed, the PDF of $g_{\mathrm p}$ is given by $f_{g_{\mathrm p}}(x) = \exp\left( -x/\Omega_{\mathrm p}\right)/\Omega_{\mathrm p}$. On the other hand, the PDF and the CDF for the random variable $g_{u'} \triangleq \sum_{i = 1}^{N_{u'}} |h_{u', i}|^2, u' \in \{{\mathrm n}, {\mathrm f}\}$, are respectively given by $f_{g_{u'}} (x) = x^{N_{u'} - 1} \exp \left( -x / \Omega_{u'}\right) / [\Gamma(N_{u'}) \Omega_{u'}^{N_{u'}}]$ and $F_{g_{u'}}(x) = 1 - \exp \left( -x / \Omega_{u'} \right) \sum_{\tau = 0}^{N_{u'} - 1} x^{\tau}/[\Omega_{u'}^{\tau}\tau !]$. The PDF of the random variable $g_{\mathrm s} = \max \{g_{\mathrm n}, g_{\mathrm f}\}$ is given by  $f_{g_{\mathrm s}}(x) = f_{g_{\mathrm n}}(x) F_{g_{\mathrm f}}(x) + f_{g_{\mathrm f}}(x) F_{g_{\mathrm n}}(x) = x^{N_{\mathrm n} - 1} \exp \left( -x / \Omega_{\mathrm n} \right)/[\Gamma(N_{\mathrm n}) \Omega_{\mathrm n}^{N_{\mathrm n}}] - \big(1/[\Gamma(N_{\mathrm n}) \Omega_{\mathrm n}^{N_{\mathrm n}}] \big) \exp \left(-x/\Omega \right) \sum_{l = 0}^{N_{\mathrm f} - 1} x^{N_{\mathrm n} + l - 1}/[l! \Omega_{\mathrm f}^l]$ $+x^{N_{\mathrm f} - 1}\exp \left( -x / \Omega_{\mathrm f} \right)/[\Gamma(N_{\mathrm f}) \Omega_{\mathrm f}^{N_{\mathrm f}}] - \big(1/[\Gamma(N_{\mathrm f}) \Omega_{\mathrm f}^{N_{\mathrm f}}] \big) $ $\exp \left(-x/\Omega \right) \sum_{\ell = 0}^{N_{\mathrm n} - 1} x^{N_{\mathrm f} + \ell - 1}/[\ell! \Omega_{\mathrm n}^{\ell}]$. Now the PDF of the random variable $X_{\mathrm s} = g_{\mathrm s}/g_{\mathrm p}$ is given by $f_{X_{\mathrm s}}(x) = \int_{0}^\infty y f_{g_{\mathrm s}}(yx) f_{g_{\mathrm p}}(y) \mathrm dy =$ $\tfrac{\Gamma(N_{\mathrm n} + 1) x^{N_{\mathrm n} - 1} }{\Gamma(N_{\mathrm n}) \Omega_{\mathrm n}^{N_{\mathrm n}} \Omega_{\mathrm p}}\left( \!\tfrac{x}{\Omega_{\mathrm n}}\! +\! \tfrac{1}{\Omega_{\mathrm p}} \!\right)^{-(N_{\mathrm n} + 1)} - \sum_{l = 0}^{N_{\mathrm f} - 1} \tfrac{\Gamma(N_{\mathrm n} + l + 1) x^{N_{\mathrm n} + l - 1}}{\Gamma(N_{\mathrm n}) \Omega_{\mathrm n}^{N_{\mathrm n} \Omega_{\mathrm p} l!} \Omega_{\mathrm f}^l}$ $\left( \tfrac{x}{\Omega} + \tfrac{1}{\Omega_{\mathrm p}} \right)^{-(N_{\mathrm n} + l + 1)} +  \tfrac{\Gamma(N_{\mathrm f} + 1) x^{N_{\mathrm f} - 1} }{\Gamma(N_{\mathrm f}) \Omega_{\mathrm f}^{N_{\mathrm f}} \Omega_{\mathrm p}}$ $\left( \tfrac{x}{\Omega_{\mathrm f}} + \tfrac{1}{\Omega_{\mathrm p}} \right)^{-(N_{\mathrm f} + 1)} - \sum_{\ell = 0}^{N_{\mathrm n} - 1} \tfrac{\Gamma(N_{\mathrm f} + \ell + 1) x^{N_{\mathrm f} + l - 1}}{\Gamma(N_{\mathrm f}) \Omega_{\mathrm f}^{N_{\mathrm f}} \Omega_{\mathrm p} \ell! \Omega_{\mathrm n}^{\ell} } \left( \tfrac{x}{\Omega} + \tfrac{1}{\Omega_{\mathrm p}} \right)^{-(N_{\mathrm f} + \ell + 1)}$.
Define $T_1(\alpha, \beta) \triangleq \int_0^\infty (1 + a_{\mathrm s}  \hat I x)^{-\nu} \tfrac{\Gamma(\alpha + 1) x^{\alpha - 1} }{\Gamma(\alpha) \beta^{\alpha} \Omega_{\mathrm p}} \left( \tfrac{x}{\beta} + \tfrac{1}{\Omega_{\mathrm p}} \right)^{-(\alpha + 1)}$ \\$ \mathrm dx =  \tfrac{\left( \Omega_{\mathrm p}/\beta\right)^{\alpha}}{\Gamma(\alpha) \Gamma(\nu)} \int_0^\infty x^{\alpha - 1} G_{1, 1}^{1, 1} \left(\! a_{\mathrm s}  \hat I x \left\vert \begin{smallmatrix} 1 - \nu \\ 0 \end{smallmatrix} \right. \!\right) G_{1, 1}^{1, 1} \left(\! \tfrac{\Omega_{\mathrm p}}{\beta} x \left\vert \begin{smallmatrix} -\alpha \\ 0 \end{smallmatrix} \right. \!\right) \mathrm dx$ $= \tfrac{1}{\Gamma(\alpha) \Gamma(\nu) (a_{\mathrm s}  \hat I)^\alpha} \left( \tfrac{\Omega_{\mathrm p}}{\beta}\right)^\alpha G_{2, 2}^{2, 2} \left( \tfrac{\Omega_{\mathrm p}}{\beta a_{\mathrm s}  \hat I} \left\vert \begin{smallmatrix} -\alpha, & 1 - \alpha \\ 0, & \nu - \alpha \end{smallmatrix} \right. \right)$,~where~the integration above is solved using~\cite[eqns.~(7),~(10),~(21),~and~(22)]{Reduce}. Similarly, define
$T_2(\alpha_1 , \alpha_2, \beta_1, \beta_2) \triangleq \int_0^\infty (1 + a_{\mathrm s}  \hat I x)^{-\nu}$ $ \sum_{\tau = 0}^{\alpha_2 - 1} \tfrac{\Gamma(\alpha_1 + \tau + 1) x^{\alpha_1 + \tau - 1}}{\Gamma(\alpha_1) \beta_1^{\alpha_1} \Omega_{\mathrm p} \beta_2^\tau \tau!} \left( \tfrac{x}{\Omega} + \tfrac{1}{\Omega_{\mathrm p}} \right)^{-(\alpha_1 + \tau + 1)} \mathrm dx = \sum_{\tau = 0}^{\alpha_2 - 1} \tfrac{[\Omega_{\mathrm p}/(a_{\mathrm s}  \hat I)]^{\alpha_1 + \tau} }{\Gamma(\alpha_1) \Gamma(\nu) \beta_1^{\alpha_1} \beta_2^{\tau} \tau!} G_{2, 2}^{2, 2} \left( \tfrac{\Omega_p}{\Omega a_{\mathrm s}  \hat I} \left\vert \begin{smallmatrix} -\alpha_1 - \tau, & 1 - \alpha_1 - \tau \\ 0, & \nu - \alpha_1 - \tau \end{smallmatrix}\right. \right)$,
where the integral above is solved similar to that for~$T_1(\alpha, \beta)$. Using~\eqref{Eqn:CaseII:Rs_def}, and the expressions for~$T_1(\alpha, \beta)$ and~$T_2(\alpha_1 , \alpha_2, \beta_1, \beta_2)$, an analytical expression for the ER of the strong user is given by~\eqref{Eqn:CaseII:RsClosed}; this completes the proof. 
%=================================
\balance
\section{Proof of Theorem~\ref{CaseII:Theorem_Rw}} \label{CaseII:Proof_Rw}
Using the expression for $f_{g_{u'}}(x)$ (given in~Appendix~\ref{CaseII:Proof_Rs}), the PDF of $g_{\mathrm w} = \min\{g_{\mathrm n}, g_{\mathrm f}\}$ is given by $f_{g_{\mathrm w}}(x) = f_{g_{\mathrm n}}(x) [1 - F_{g_{\mathrm f}}(x)] + f_{g_{\mathrm f}}(x) [1 - F_{g_{\mathrm n}}(x)] = \sum_{l = 0}^{N_{\mathrm f} - 1} \tfrac{x^{N_{\mathrm n} + l - 1} \exp \left(-x/\Omega\right)}{\Gamma(N_{\mathrm n}) \Omega_{\mathrm n}^{N_{\mathrm n}} \Omega_{\mathrm f}^{l} l!}  +  \sum_{\ell = 0}^{N_{\mathrm n} - 1} \tfrac{x^{N_{\mathrm f} + \ell - 1} \exp \left( -x/\Omega\right)}{\Gamma(N_{\mathrm f}) \Omega_{\mathrm f}^{N_{\mathrm f}} \Omega_{\mathrm n}^{\ell} \ell!}$.
Now the PDF of $X_{\mathrm w}$ can be given by $f_{X_{\mathrm w}}(x) = \int_0^\infty y f_{g_{\mathrm w}}(yx) f_{g_{\mathrm p}}(y) \mathrm dy = \sum_{l = 0}^{N_{\mathrm f} - 1} \tfrac{\Gamma(N_{\mathrm n} + l + 1)}{\Gamma(N_{\mathrm n}) \Omega_{\mathrm n}^{N_{\mathrm n}} \Omega_{\mathrm f}^l \Omega_{\mathrm p} l!} x^{N_{\mathrm n} + l - 1} \left( \tfrac{x}{\Omega} + \tfrac{1}{\Omega_{\mathrm p}}\right)^{-(N_{\mathrm n} + l + 1)} + \sum_{\ell = 0}^{N_{\mathrm n} - 1} \tfrac{\Gamma(N_{\mathrm f} + \ell + 1)}{\Gamma(N_{\mathrm f}) \Omega_{\mathrm f}^{N_{\mathrm f}} \Omega_{\mathrm n}^{\ell} \Omega_{\mathrm p} \ell!} x^{N_{\mathrm f} + \ell - 1} \left( \tfrac{x}{\Omega} + \tfrac{1}{\Omega_{\mathrm p}}\right)^{-(N_{\mathrm f} + \ell + 1)}$.
%\begin{align}
%	& f_{X_{\min}}(x) = \int_0^\infty y f_{g_{\min}}(yx) f_{g_p}(y) \mathrm dy \notag \\
%	= & \ \sum_{k = 0}^{N_f - 1} \dfrac{\Gamma(N_n + k + 1)}{\Gamma(N_n) \Omega_n^{N_n} \Omega_f^k \Omega_p k!} x^{N_n + k - 1} \left( \dfrac{x}{\Omega} + \dfrac{1}{\Omega_p}\right)^{-(N_n + k + 1)} \notag \\
%	& + \sum_{l = 0}^{N_n - 1} \dfrac{\Gamma(N_f + l + 1)}{\Gamma(N_f) \Omega_f^{N_f} \Omega_n^l \Omega_p l!} x^{N_f + l - 1} \left( \dfrac{x}{\Omega} + \dfrac{1}{\Omega_p}\right)^{-(N_f + l + 1)}. \label{Eqn:f_Xmin}
%\end{align}
Define
$T_3(\alpha_1, \alpha_2, \beta_1, \beta_2) \triangleq \sum_{\tau = 0}^{\alpha_2 - 1} \tfrac{\Gamma(\alpha_1 + \tau + 1)}{\Gamma(\alpha_1) \beta_1^{\alpha_1} \beta_2^{\tau} \Omega_{\mathrm p} \tau!} \int_0^\infty x^{\alpha_1 + \tau - 1}$ $(1 + \tfrac{2}{K}Ix)^{-\nu} (1 + a_{\mathrm s}  \hat I x)^{\nu} \left( \tfrac{x}{\Omega} + \tfrac{1}{\Omega_{\mathrm p}}\right)^{-(\alpha_1 + \tau + 1)} \mathrm dx = \sum_{\tau = 0}^{\alpha_2 - 1} \tfrac{\Omega_{\mathrm p}^{\alpha_1 + \tau}}{\Gamma(\alpha_1) \tau! \beta_1^{\alpha_1} \beta_2^\tau }$ $\tfrac{1}{\Gamma(\nu) \Gamma(-\nu)} \int_0^\infty x^{\alpha_1 + \tau -1}  G_{1, 1}^{1, 1} \left( \tfrac{2}{K}  \hat I x \left\vert \begin{smallmatrix} 1 - \nu \\ 0\end{smallmatrix}\right. \right) G_{1, 1}^{1, 1} \left( a_{\mathrm s}  \hat I x \left\vert \begin{smallmatrix} 1 + \nu \\ 0 \end{smallmatrix}\right. \right)$ $G_{1, 1}^{1, 1} \left( \tfrac{\Omega_{\mathrm p}}{\Omega} x \left\vert \begin{smallmatrix} -(\alpha_1 + \tau) \\ 0\end{smallmatrix}\right. \right) \mathrm dx
	= \sum_{\tau = 0}^{\alpha_2 - 1} \tfrac{\Omega_{\mathrm p} (2 \hat I/K)^{-(\alpha_1 + \tau)}}{\Gamma(\alpha_1) \Gamma(\nu) \Gamma(-\nu)  \beta_1^{\alpha_1} \beta_2^{\tau} \tau!}$ $\mathcal G_{1, 1:1, 1:1, 1}^{1, 1:1, 1:1, 1} \left(\begin{smallmatrix} 1 - (\alpha_1 + \tau) \\ \nu - (\alpha_1 + \tau)\end{smallmatrix} \left\vert \begin{smallmatrix} 1 - \nu \\ 0 \end{smallmatrix} \right\vert \left. \begin{smallmatrix} -(\alpha_1 + \tau) \\ 0 \end{smallmatrix} \right\vert \tfrac{K}{2}a_{\mathrm s}, \tfrac{K\Omega_{\mathrm p}}{2\Omega  \hat I} \right)$, where the integral above is solved using~\cite[eqn.~(10)]{Reduce} and~\cite[eqn.~(9)]{Hindawi}. Using the expressions for $f_{X_{\mathrm w}}(x)$ and~$T_3(\alpha_1, \alpha_2, \beta_1, \beta_2)$, and~\eqref{Eqn:CaseII:Rw_def}, an analytical expression for the ER of the weak user is given by~\eqref{Eqn:CaseII:RwClosed}; this concludes the proof. 
%===============References========
%\vspace{-0.3cm}
\bibliographystyle{IEEEtran}
\bibliography{SSEC}
\end{document}

%% file: Fig1a.tex
%% II Case
\pgfplotsset{compat=1.9}
\resizebox{\columnwidth}{!}{
\begin{tikzpicture}   
      \begin{axis}[
      	xlabel={$I$ (dB)},
	xtick={-20, -16, -12, -8, -4, 0},
	xmin=-20,xmax=0,
	ylabel={Sum ER (bps/Hz)},
	ymin=0,ymax=10,
	grid=both,
	minor grid style={gray!25},
	major grid style={gray!25},
	legend columns=2, 
legend style={{nodes={scale=1, transform shape}}, at={(0.001,0.716)},  anchor=south west, draw=black,fill=white,legend cell align=left,inner sep=1pt,row sep = -2pt},]
	\addplot[mark=triangle*, only marks, mark size=2.5pt, color=cyan] table [y=NOMA_4_4, x=I_dB,col sep = comma]{Data_Rate_vs_IdB_II.csv};
	\addlegendentry{NOMA (4, 4)}
	\addplot[color=cyan,dashed, line width=0.5pt, mark size=2.5pt, mark=triangle, mark options={solid, cyan}] table [y=OMA_4_4, x=I_dB,col sep = comma]{Data_Rate_vs_IdB_II.csv};
	\addlegendentry{OMA (4, 4)}
	\addplot[mark=square*, only marks, mark size=1.5pt, color=violet] table [y=NOMA_4_1, x=I_dB,col sep = comma]{Data_Rate_vs_IdB_II.csv};
	\addlegendentry{NOMA (4, 1)}
	\addplot[color=violet,dashed, line width=0.5pt, mark size=1.5pt, mark=square, mark options={solid, violet}] table [y=OMA_4_1, x=I_dB,col sep = comma]{Data_Rate_vs_IdB_II.csv};
	\addlegendentry{OMA (4, 1)}
	\addplot[mark=diamond*, only marks, mark size=2.5pt, color=green] table [y=NOMA_2_4, x=I_dB,col sep = comma]{Data_Rate_vs_IdB_II.csv};
	\addlegendentry{NOMA (2, 4)}
	\addplot[color=green,dashed, line width=0.5pt, mark size=2.5pt, mark=diamond, mark options={solid, green}] table [y=OMA_2_4, x=I_dB,col sep = comma]{Data_Rate_vs_IdB_II.csv};
	\addlegendentry{OMA (2, 4)}
	\addplot[mark=*, only marks, mark size=1.5pt, color=red] table [y=NOMA_2_1, x=I_dB,col sep = comma]{Data_Rate_vs_IdB_II.csv};
	\addlegendentry{NOMA (2, 1)}
	\addplot[color=red,dashed, line width=0.5pt, mark size=1.5pt, mark=o, mark options={solid, red}] table [y=OMA_2_1, x=I_dB,col sep = comma]{Data_Rate_vs_IdB_II.csv};
	\addlegendentry{OMA (2, 1)}
	\addplot[line width=0.5pt,solid,color=red] table [y=NOMA_2_1, x=I_dB,col sep = comma]{Data_Rate_vs_IdB_II.csv};
	\addplot[line width=0.5pt,solid,color=green] table [y=NOMA_2_4, x=I_dB,col sep = comma]{Data_Rate_vs_IdB_II.csv};
	\addplot[line width=0.5pt,solid,color=violet] table [y=NOMA_4_1, x=I_dB,col sep = comma]{Data_Rate_vs_IdB_II.csv};
	\addplot[line width=0.5pt,solid,color=cyan] table [y=NOMA_4_4, x=I_dB,col sep = comma]{Data_Rate_vs_IdB_II.csv};
	\end{axis}
\end{tikzpicture}
}

%% file: Fig1b.tex
%% IS Case
\pgfplotsset{compat=1.9}
\resizebox{\columnwidth}{!}{
\begin{tikzpicture}   
      \begin{axis}[
      	xlabel={$I$ (dB)},
	xtick={-20, -16, -12, -8, -4, 0},
	xmin=-20,xmax=0,
	ylabel={Sum ER (bps/Hz)},
	ymin=0,ymax=10,
	grid=both,
	minor grid style={gray!25},
	major grid style={gray!25},
	legend columns=2, 
legend style={{nodes={scale=1, transform shape}}, at={(0.001,0.716)},  anchor=south west, draw=black,fill=white,legend cell align=left,inner sep=1pt,row sep = -2pt},]
	\addplot[mark=triangle*, only marks, mark size=2.5pt, color=cyan] table [y=NOMA_4_4, x=I_dB,col sep = comma]{Data_Rate_vs_IdB_IS.csv};
	\addlegendentry{NOMA (4, 4)}
	\addplot[color=cyan,dashed, line width=0.5pt, mark size=2.5pt, mark=triangle, mark options={solid, cyan}] table [y=OMA_4_4, x=I_dB,col sep = comma]{Data_Rate_vs_IdB_IS.csv};
	\addlegendentry{OMA (4, 4)}
	\addplot[mark=square*, only marks, mark size=1.5pt, color=violet] table [y=NOMA_4_1, x=I_dB,col sep = comma]{Data_Rate_vs_IdB_IS.csv};
	\addlegendentry{NOMA (4, 1)}
	\addplot[color=violet,dashed, line width=0.5pt, mark size=1.5pt, mark=square, mark options={solid, violet}] table [y=OMA_4_1, x=I_dB,col sep = comma]{Data_Rate_vs_IdB_IS.csv};
	\addlegendentry{OMA (4, 1)}
	\addplot[mark=diamond*, only marks, mark size=2.5pt, color=green] table [y=NOMA_2_4, x=I_dB,col sep = comma]{Data_Rate_vs_IdB_IS.csv};
	\addlegendentry{NOMA (2, 4)}
	\addplot[color=green,dashed, line width=0.5pt, mark size=2.5pt, mark=diamond, mark options={solid, green}] table [y=OMA_2_4, x=I_dB,col sep = comma]{Data_Rate_vs_IdB_IS.csv};
	\addlegendentry{OMA (2, 4)}
	\addplot[mark=*, only marks, mark size=1.5pt, color=red] table [y=NOMA_2_1, x=I_dB,col sep = comma]{Data_Rate_vs_IdB_IS.csv};
	\addlegendentry{NOMA (2, 1)}
	\addplot[color=red,dashed, line width=0.5pt, mark size=1.5pt, mark=o, mark options={solid, red}] table [y=OMA_2_1, x=I_dB,col sep = comma]{Data_Rate_vs_IdB_IS.csv};
	\addlegendentry{OMA (2, 1)}
	\addplot[line width=0.5pt,solid,color=red] table [y=NOMA_2_1, x=I_dB,col sep = comma]{Data_Rate_vs_IdB_IS.csv};
	\addplot[line width=0.5pt,solid,color=green] table [y=NOMA_2_4, x=I_dB,col sep = comma]{Data_Rate_vs_IdB_IS.csv};
	\addplot[line width=0.5pt,solid,color=violet] table [y=NOMA_4_1, x=I_dB,col sep = comma]{Data_Rate_vs_IdB_IS.csv};
	\addplot[line width=0.5pt,solid,color=cyan] table [y=NOMA_4_4, x=I_dB,col sep = comma]{Data_Rate_vs_IdB_IS.csv};
	\end{axis}
\end{tikzpicture}
}

%% file: Fig1c.tex
%% SI Case
\pgfplotsset{compat=1.9}
\resizebox{\columnwidth}{!}{
\begin{tikzpicture}   
      \begin{axis}[
      	xlabel={$I$ (dB)},
	xtick={-20, -16, -12, -8, -4, 0},
	xmin=-20,xmax=0,
	ylabel={Sum ER (bps/Hz)},
	ymin=0,ymax=10,
	grid=both,
	minor grid style={gray!25},
	major grid style={gray!25},
	legend columns=2, 
legend style={{nodes={scale=1, transform shape}}, at={(0.001,0.716)},  anchor=south west, draw=black,fill=white,legend cell align=left,inner sep=1pt,row sep = -2pt},]
	\addplot[mark=triangle*, only marks, mark size=2.5pt, color=cyan] table [y=NOMA_4_4, x=I_dB,col sep = comma]{Data_Rate_vs_IdB_SI.csv};
	\addlegendentry{NOMA (4, 4)}
	\addplot[color=cyan,dashed, line width=0.5pt, mark size=2.5pt, mark=triangle, mark options={solid, cyan}] table [y=OMA_4_4, x=I_dB,col sep = comma]{Data_Rate_vs_IdB_SI.csv};
	\addlegendentry{OMA (4, 4)}
	\addplot[mark=square*, only marks, mark size=1.5pt, color=violet] table [y=NOMA_4_1, x=I_dB,col sep = comma]{Data_Rate_vs_IdB_SI.csv};
	\addlegendentry{NOMA (4, 1)}
	\addplot[color=violet,dashed, line width=0.5pt, mark size=1.5pt, mark=square, mark options={solid, violet}] table [y=OMA_4_1, x=I_dB,col sep = comma]{Data_Rate_vs_IdB_SI.csv};
	\addlegendentry{OMA (4, 1)}
	\addplot[mark=diamond*, only marks, mark size=2.5pt, color=green] table [y=NOMA_2_4, x=I_dB,col sep = comma]{Data_Rate_vs_IdB_SI.csv};
	\addlegendentry{NOMA (2, 4)}
	\addplot[color=green,dashed, line width=0.5pt, mark size=2.5pt, mark=diamond, mark options={solid, green}] table [y=OMA_2_4, x=I_dB,col sep = comma]{Data_Rate_vs_IdB_SI.csv};
	\addlegendentry{OMA (2, 4)}
	\addplot[mark=*, only marks, mark size=1.5pt, color=red] table [y=NOMA_2_1, x=I_dB,col sep = comma]{Data_Rate_vs_IdB_SI.csv};
	\addlegendentry{NOMA (2, 1)}
	\addplot[color=red,dashed, line width=0.5pt, mark size=1.5pt, mark=o, mark options={solid, red}] table [y=OMA_2_1, x=I_dB,col sep = comma]{Data_Rate_vs_IdB_SI.csv};
	\addlegendentry{OMA (2, 1)}
	\addplot[line width=0.5pt,solid,color=red] table [y=NOMA_2_1, x=I_dB,col sep = comma]{Data_Rate_vs_IdB_SI.csv};
	\addplot[line width=0.5pt,solid,color=green] table [y=NOMA_2_4, x=I_dB,col sep = comma]{Data_Rate_vs_IdB_SI.csv};
	\addplot[line width=0.5pt,solid,color=violet] table [y=NOMA_4_1, x=I_dB,col sep = comma]{Data_Rate_vs_IdB_SI.csv};
	\addplot[line width=0.5pt,solid,color=cyan] table [y=NOMA_4_4, x=I_dB,col sep = comma]{Data_Rate_vs_IdB_SI.csv};
	\end{axis}
\end{tikzpicture}
}

%% file: Fig1d.tex
%% SS Case
\pgfplotsset{compat=1.9}
\resizebox{\columnwidth}{!}{
\begin{tikzpicture}   
      \begin{axis}[
      	xlabel={$I$ (dB)},
	xtick={-20, -16, -12, -8, -4, 0},
	xmin=-20,xmax=0,
	ylabel={Sum ER (bps/Hz)},
	ymin=0,ymax=10,
	grid=both,
	minor grid style={gray!25},
	major grid style={gray!25},
	legend columns=2, 
legend style={{nodes={scale=1, transform shape}}, at={(0.001,0.716)},  anchor=south west, draw=black,fill=white,legend cell align=left,inner sep=1pt,row sep = -2pt},]
	\addplot[mark=triangle*, only marks, mark size=2.5pt, color=cyan] table [y=NOMA_4_4, x=I_dB,col sep = comma]{Data_Rate_vs_IdB_SS.csv};
	\addlegendentry{NOMA (4, 4)}
	\addplot[color=cyan,dashed, line width=0.5pt, mark size=2.5pt, mark=triangle, mark options={solid, cyan}] table [y=OMA_4_4, x=I_dB,col sep = comma]{Data_Rate_vs_IdB_SS.csv};
	\addlegendentry{OMA (4, 4)}
	\addplot[mark=square*, only marks, mark size=1.5pt, color=violet] table [y=NOMA_4_1, x=I_dB,col sep = comma]{Data_Rate_vs_IdB_SS.csv};
	\addlegendentry{NOMA (4, 1)}
	\addplot[color=violet,dashed, line width=0.5pt, mark size=1.5pt, mark=square, mark options={solid, violet}] table [y=OMA_4_1, x=I_dB,col sep = comma]{Data_Rate_vs_IdB_SS.csv};
	\addlegendentry{OMA (4, 1)}
	\addplot[mark=diamond*, only marks, mark size=2.5pt, color=green] table [y=NOMA_2_4, x=I_dB,col sep = comma]{Data_Rate_vs_IdB_SS.csv};
	\addlegendentry{NOMA (2, 4)}
	\addplot[color=green,dashed, line width=0.5pt, mark size=2.5pt, mark=diamond, mark options={solid, green}] table [y=OMA_2_4, x=I_dB,col sep = comma]{Data_Rate_vs_IdB_SS.csv};
	\addlegendentry{OMA (2, 4)}
	\addplot[mark=*, only marks, mark size=1.5pt, color=red] table [y=NOMA_2_1, x=I_dB,col sep = comma]{Data_Rate_vs_IdB_SS.csv};
	\addlegendentry{NOMA (2, 1)}
	\addplot[color=red,dashed, line width=0.5pt, mark size=1.5pt, mark=o, mark options={solid, red}] table [y=OMA_2_1, x=I_dB,col sep = comma]{Data_Rate_vs_IdB_SS.csv};
	\addlegendentry{OMA (2, 1)}
	\addplot[line width=0.5pt,solid,color=red] table [y=NOMA_2_1, x=I_dB,col sep = comma]{Data_Rate_vs_IdB_SS.csv};
	\addplot[line width=0.5pt,solid,color=green] table [y=NOMA_2_4, x=I_dB,col sep = comma]{Data_Rate_vs_IdB_SS.csv};
	\addplot[line width=0.5pt,solid,color=violet] table [y=NOMA_4_1, x=I_dB,col sep = comma]{Data_Rate_vs_IdB_SS.csv};
	\addplot[line width=0.5pt,solid,color=cyan] table [y=NOMA_4_4, x=I_dB,col sep = comma]{Data_Rate_vs_IdB_SS.csv};
	\end{axis}
\end{tikzpicture}
}	

%% file: Fig2a.tex
%% II Case
\pgfplotsset{compat=1.9}
\resizebox{\columnwidth}{!}{
\begin{tikzpicture}   
      \begin{axis}[
      	xmode=log,
	log basis x={10},
      	xlabel={Delay exponent $\theta$},
	xtick={0.01, 0.1, 1, 10, 100, 1000},
	xmin=0.01,xmax=1000,
	ylabel={Sum ER (bps/Hz)},
	ymin=0,ymax=11,
	grid=both,
	minor grid style={gray!25},
	major grid style={gray!25},
	legend columns=1, 
legend style={{nodes={scale=0.95, transform shape}}, at={(0.6,0.516)},  anchor=south west, draw=black,fill=white,legend cell align=left,inner sep=1pt,row sep = -2.5pt},]
	\addplot[mark=triangle*, only marks, mark size=2.5pt, color=cyan] table [y=NOMA_4_4, x=theta,col sep = comma]{Data_Rate_vs_theta_II.csv};
	\addlegendentry{NOMA (4, 4)}
	\addplot[color=cyan,dashed, line width=0.5pt, mark size=2.5pt, mark=triangle, mark options={solid, cyan}] table [y=OMA_4_4, x=theta,col sep = comma]{Data_Rate_vs_theta_II.csv};
	\addlegendentry{OMA (4, 4)}
	\addplot[mark=square*, only marks, mark size=1.5pt, color=violet] table [y=NOMA_4_1, x=theta,col sep = comma]{Data_Rate_vs_theta_II.csv};
	\addlegendentry{NOMA (4, 1)}
	\addplot[color=violet,dashed, line width=0.5pt, mark size=1.5pt, mark=square, mark options={solid, violet}] table [y=OMA_4_1, x=theta,col sep = comma]{Data_Rate_vs_theta_II.csv};
	\addlegendentry{OMA (4, 1)}
	\addplot[mark=diamond*, only marks, mark size=2.5pt, color=green] table [y=NOMA_2_4, x=theta,col sep = comma]{Data_Rate_vs_theta_II.csv};
	\addlegendentry{NOMA (2, 4)}
	\addplot[color=green,dashed, line width=0.5pt, mark size=2.5pt, mark=diamond, mark options={solid, green}] table [y=OMA_2_4, x=theta,col sep = comma]{Data_Rate_vs_theta_II.csv};
	\addlegendentry{OMA (2, 4)}
	\addplot[mark=*, only marks, mark size=1.5pt, color=red] table [y=NOMA_2_1, x=theta,col sep = comma]{Data_Rate_vs_theta_II.csv};
	\addlegendentry{NOMA (2, 1)}
	\addplot[color=red,dashed, line width=0.5pt, mark size=1.5pt, mark=o, mark options={solid, red}] table [y=OMA_2_1, x=theta,col sep = comma]{Data_Rate_vs_theta_II.csv};
	\addlegendentry{OMA (2, 1)}
	\addplot[line width=0.5pt,solid,color=red] table [y=NOMA_2_1, x=theta,col sep = comma]{Data_Rate_vs_theta_II.csv};
	\addplot[line width=0.5pt,solid,color=green] table [y=NOMA_2_4, x=theta,col sep = comma]{Data_Rate_vs_theta_II.csv};
	\addplot[line width=0.5pt,solid,color=violet] table [y=NOMA_4_1, x=theta,col sep = comma]{Data_Rate_vs_theta_II.csv};
	\addplot[line width=0.5pt,solid,color=cyan] table [y=NOMA_4_4, x=theta,col sep = comma]{Data_Rate_vs_theta_II.csv};
	\end{axis}
\end{tikzpicture}
}

%% file: Fig2b.tex
%% II Case
\pgfplotsset{compat=1.9}
\resizebox{\columnwidth}{!}{
\begin{tikzpicture}   
      \begin{axis}[
      	xmode=log,
	log basis x={10},
      	xlabel={Delay exponent $\theta$},
	xtick={0.01, 0.1, 1, 10, 100, 1000},
	xmin=0.01,xmax=1000,
	ylabel={Sum ER (bps/Hz)},
	ymin=0,ymax=11,
	grid=both,
	minor grid style={gray!25},
	major grid style={gray!25},
	legend columns=1, 
legend style={{nodes={scale=0.95, transform shape}}, at={(0.6,0.516)},  anchor=south west, draw=black,fill=white,legend cell align=left,inner sep=1pt,row sep = -2.5pt}]
	\addplot[mark=triangle*, only marks, mark size=2.5pt, color=cyan] table [y=NOMA_4_4, x=theta,col sep = comma]{Data_Rate_vs_theta_IS.csv};
	\addlegendentry{NOMA (4, 4)}
	\addplot[color=cyan,dashed, line width=0.5pt, mark size=2.5pt, mark=triangle, mark options={solid, cyan}] table [y=OMA_4_4, x=theta,col sep = comma]{Data_Rate_vs_theta_IS.csv};
	\addlegendentry{OMA (4, 4)}
	\addplot[mark=square*, only marks, mark size=1.5pt, color=violet] table [y=NOMA_4_1, x=theta,col sep = comma]{Data_Rate_vs_theta_IS.csv};
	\addlegendentry{NOMA (4, 1)}
	\addplot[color=violet,dashed, line width=0.5pt, mark size=1.5pt, mark=square, mark options={solid, violet}] table [y=OMA_4_1, x=theta,col sep = comma]{Data_Rate_vs_theta_IS.csv};
	\addlegendentry{OMA (4, 1)}
	\addplot[mark=diamond*, only marks, mark size=2.5pt, color=green] table [y=NOMA_2_4, x=theta,col sep = comma]{Data_Rate_vs_theta_IS.csv};
	\addlegendentry{NOMA (2, 4)}
	\addplot[color=green,dashed, line width=0.5pt, mark size=2.5pt, mark=diamond, mark options={solid, green}] table [y=OMA_2_4, x=theta,col sep = comma]{Data_Rate_vs_theta_IS.csv};
	\addlegendentry{OMA (2, 4)}
	\addplot[mark=*, only marks, mark size=1.5pt, color=red] table [y=NOMA_2_1, x=theta,col sep = comma]{Data_Rate_vs_theta_IS.csv};
	\addlegendentry{NOMA (2, 1)}
	\addplot[color=red,dashed, line width=0.5pt, mark size=1.5pt, mark=o, mark options={solid, red}] table [y=OMA_2_1, x=theta,col sep = comma]{Data_Rate_vs_theta_IS.csv};
	\addlegendentry{OMA (2, 1)}
	\addplot[line width=0.5pt,solid,color=red] table [y=NOMA_2_1, x=theta,col sep = comma]{Data_Rate_vs_theta_IS.csv};
	\addplot[line width=0.5pt,solid,color=green] table [y=NOMA_2_4, x=theta,col sep = comma]{Data_Rate_vs_theta_IS.csv};
	\addplot[line width=0.5pt,solid,color=violet] table [y=NOMA_4_1, x=theta,col sep = comma]{Data_Rate_vs_theta_IS.csv};
	\addplot[line width=0.5pt,solid,color=cyan] table [y=NOMA_4_4, x=theta,col sep = comma]{Data_Rate_vs_theta_IS.csv};
	\end{axis}
\end{tikzpicture}
}

%% file: Fig2c.tex
%% II Case
\pgfplotsset{compat=1.9}
\resizebox{\columnwidth}{!}{
\begin{tikzpicture}   
      \begin{axis}[
      	xmode=log,
	log basis x={10},
      	xlabel={Delay exponent $\theta$},
	xtick={0.01, 0.1, 1, 10, 100, 1000},
	xmin=0.01,xmax=1000,
	ylabel={Sum ER (bps/Hz)},
	ymin=0,ymax=11,
	grid=both,
	minor grid style={gray!25},
	major grid style={gray!25},
	legend columns=1, 
legend style={{nodes={scale=0.95, transform shape}}, at={(0.6,0.516)},  anchor=south west, draw=black,fill=white,legend cell align=left,inner sep=1pt,row sep = -2.5pt}]
	\addplot[mark=triangle*, only marks, mark size=2.5pt, color=cyan] table [y=NOMA_4_4, x=theta,col sep = comma]{Data_Rate_vs_theta_SI.csv};
	\addlegendentry{NOMA (4, 4)}
	\addplot[color=cyan,dashed, line width=0.5pt, mark size=2.5pt, mark=triangle, mark options={solid, cyan}] table [y=OMA_4_4, x=theta,col sep = comma]{Data_Rate_vs_theta_SI.csv};
	\addlegendentry{OMA (4, 4)}
	\addplot[mark=square*, only marks, mark size=1.5pt, color=violet] table [y=NOMA_4_1, x=theta,col sep = comma]{Data_Rate_vs_theta_SI.csv};
	\addlegendentry{NOMA (4, 1)}
	\addplot[color=violet,dashed, line width=0.5pt, mark size=1.5pt, mark=square, mark options={solid, violet}] table [y=OMA_4_1, x=theta,col sep = comma]{Data_Rate_vs_theta_SI.csv};
	\addlegendentry{OMA (4, 1)}
	\addplot[mark=diamond*, only marks, mark size=2.5pt, color=green] table [y=NOMA_2_4, x=theta,col sep = comma]{Data_Rate_vs_theta_SI.csv};
	\addlegendentry{NOMA (2, 4)}
	\addplot[color=green,dashed, line width=0.5pt, mark size=2.5pt, mark=diamond, mark options={solid, green}] table [y=OMA_2_4, x=theta,col sep = comma]{Data_Rate_vs_theta_SI.csv};
	\addlegendentry{OMA (2, 4)}
	\addplot[mark=*, only marks, mark size=1.5pt, color=red] table [y=NOMA_2_1, x=theta,col sep = comma]{Data_Rate_vs_theta_SI.csv};
	\addlegendentry{NOMA (2, 1)}
	\addplot[color=red,dashed, line width=0.5pt, mark size=1.5pt, mark=o, mark options={solid, red}] table [y=OMA_2_1, x=theta,col sep = comma]{Data_Rate_vs_theta_SI.csv};
	\addlegendentry{OMA (2, 1)}
	\addplot[line width=0.5pt,solid,color=red] table [y=NOMA_2_1, x=theta,col sep = comma]{Data_Rate_vs_theta_SI.csv};
	\addplot[line width=0.5pt,solid,color=green] table [y=NOMA_2_4, x=theta,col sep = comma]{Data_Rate_vs_theta_SI.csv};
	\addplot[line width=0.5pt,solid,color=violet] table [y=NOMA_4_1, x=theta,col sep = comma]{Data_Rate_vs_theta_SI.csv};
	\addplot[line width=0.5pt,solid,color=cyan] table [y=NOMA_4_4, x=theta,col sep = comma]{Data_Rate_vs_theta_SI.csv};
	\end{axis}
\end{tikzpicture}
}

%% file: Fig2d.tex
%% II Case
\pgfplotsset{compat=1.9}
\resizebox{\columnwidth}{!}{
\begin{tikzpicture}   
      \begin{axis}[
      	xmode=log,
	log basis x={10},
      	xlabel={Delay exponent $\theta$},
	xtick={0.01, 0.1, 1, 10, 100, 1000},
	xmin=0.01,xmax=1000,
	ylabel={Sum ER (bps/Hz)},
	ymin=0,ymax=11,
	grid=both,
	minor grid style={gray!25},
	major grid style={gray!25},
	legend columns=1, 
legend style={{nodes={scale=0.95, transform shape}}, at={(0.6,0.516)},  anchor=south west, draw=black,fill=white,legend cell align=left,inner sep=1pt,row sep = -2.5pt}]
	\addplot[mark=triangle*, only marks, mark size=2.5pt, color=cyan] table [y=NOMA_4_4, x=theta,col sep = comma]{Data_Rate_vs_theta_SS.csv};
	\addlegendentry{NOMA (4, 4)}
	\addplot[color=cyan,dashed, line width=0.5pt, mark size=2.5pt, mark=triangle, mark options={solid, cyan}] table [y=OMA_4_4, x=theta,col sep = comma]{Data_Rate_vs_theta_SS.csv};
	\addlegendentry{OMA (4, 4)}
	\addplot[mark=square*, only marks, mark size=1.5pt, color=violet] table [y=NOMA_4_1, x=theta,col sep = comma]{Data_Rate_vs_theta_SS.csv};
	\addlegendentry{NOMA (4, 1)}
	\addplot[color=violet,dashed, line width=0.5pt, mark size=1.5pt, mark=square, mark options={solid, violet}] table [y=OMA_4_1, x=theta,col sep = comma]{Data_Rate_vs_theta_SS.csv};
	\addlegendentry{OMA (4, 1)}
	\addplot[mark=diamond*, only marks, mark size=2.5pt, color=green] table [y=NOMA_2_4, x=theta,col sep = comma]{Data_Rate_vs_theta_SS.csv};
	\addlegendentry{NOMA (2, 4)}
	\addplot[color=green,dashed, line width=0.5pt, mark size=2.5pt, mark=diamond, mark options={solid, green}] table [y=OMA_2_4, x=theta,col sep = comma]{Data_Rate_vs_theta_SS.csv};
	\addlegendentry{OMA (2, 4)}
	\addplot[mark=*, only marks, mark size=1.5pt, color=red] table [y=NOMA_2_1, x=theta,col sep = comma]{Data_Rate_vs_theta_SS.csv};
	\addlegendentry{NOMA (2, 1)}
	\addplot[color=red,dashed, line width=0.5pt, mark size=1.5pt, mark=o, mark options={solid, red}] table [y=OMA_2_1, x=theta,col sep = comma]{Data_Rate_vs_theta_SS.csv};
	\addlegendentry{OMA (2, 1)}
	\addplot[line width=0.5pt,solid,color=red] table [y=NOMA_2_1, x=theta,col sep = comma]{Data_Rate_vs_theta_SS.csv};
	\addplot[line width=0.5pt,solid,color=green] table [y=NOMA_2_4, x=theta,col sep = comma]{Data_Rate_vs_theta_SS.csv};
	\addplot[line width=0.5pt,solid,color=violet] table [y=NOMA_4_1, x=theta,col sep = comma]{Data_Rate_vs_theta_SS.csv};
	\addplot[line width=0.5pt,solid,color=cyan] table [y=NOMA_4_4, x=theta,col sep = comma]{Data_Rate_vs_theta_SS.csv};
	\end{axis}
\end{tikzpicture}
}

%% file: Fig3.tex
\pgfplotsset{compat=1.9}
\resizebox{0.7\columnwidth}{!}{
\begin{tikzpicture}   
      \begin{axis}[
      	xlabel={$I$ (dB)},
	xtick={-20, -16, -12, -8, -4, 0},
	xmin=-20,xmax=0,
	ylabel={Sum ER (bps/Hz)},
	ymin=0,ymax=10,
	grid=both,
	minor grid style={gray!25},
	major grid style={gray!25},
	legend columns=2, 
legend style={{nodes={scale=0.8, transform shape}}, at={(0.001,0.586)},  anchor=south west, draw=black,fill=white,legend cell align=left,inner sep=1pt,row sep = -2pt},]
	%%========== 4-users, 4-antennas==============
	\addplot[mark=triangle*, only marks, mark size=1.5pt, color=orange] table [y=NOMA_4_4, x=I_dB,col sep = comma]{Data_Rate_vs_IdB_II.csv};
	\addlegendentry{I-I (4, 4)}
	\addplot[mark=triangle, only marks, mark size=3.5pt, color=orange] table [y=NOMA_4_4, x=I_dB,col sep = comma]{Data_Rate_vs_IdB_IS.csv};
	\addlegendentry{I-S (4, 4)}
	\addplot[mark=triangle*, only marks, mark size=1.5pt, color=magenta] table [y=NOMA_4_4, x=I_dB,col sep = comma]{Data_Rate_vs_IdB_SI.csv};
	\addlegendentry{S-I (4, 4)}
	\addplot[mark=triangle, only marks, mark size=3.5pt, color=magenta] table [y=NOMA_4_4, x=I_dB,col sep = comma]{Data_Rate_vs_IdB_SS.csv};
	\addlegendentry{S-S (4, 4)}
	\addplot[mark=square*, only marks, mark size=1pt, color=violet] table [y=NOMA_4_1, x=I_dB,col sep = comma]{Data_Rate_vs_IdB_II.csv};
	\addlegendentry{I-I (4, 1)}
	\addplot[mark=square, only marks, mark size=2.5pt, color=violet] table [y=NOMA_4_1, x=I_dB,col sep = comma]{Data_Rate_vs_IdB_IS.csv};
	\addlegendentry{I-S (4, 1)}
	\addplot[mark=diamond*, only marks, mark size=1.5pt, color=cyan] table [y=NOMA_2_4, x=I_dB,col sep = comma]{Data_Rate_vs_IdB_II.csv};
	\addlegendentry{I-I (2, 4)}
	\addplot[mark=diamond, only marks, mark size=3.5pt, color=cyan] table [y=NOMA_2_4, x=I_dB,col sep = comma]{Data_Rate_vs_IdB_IS.csv};
	\addlegendentry{I-S (2, 4)}
	\addplot[mark=diamond*, only marks, mark size=1.5pt, color=teal] table [y=NOMA_2_4, x=I_dB,col sep = comma]{Data_Rate_vs_IdB_SI.csv};
	\addlegendentry{S-I (2, 4)}
	\addplot[mark=diamond, only marks, mark size=3.5pt, color=teal] table [y=NOMA_2_4, x=I_dB,col sep = comma]{Data_Rate_vs_IdB_SS.csv};
	\addlegendentry{S-S (2, 4)}
	\addplot[mark=square*, only marks, mark size=1pt, color=green] table [y=NOMA_4_1, x=I_dB,col sep = comma]{Data_Rate_vs_IdB_SI.csv};
	\addlegendentry{S-I (4, 1)}
	\addplot[mark=square, only marks, mark size=2.5pt, color=green] table [y=NOMA_4_1, x=I_dB,col sep = comma]{Data_Rate_vs_IdB_SS.csv};
	\addlegendentry{S-S (4, 1)}
	\addplot[mark=*, only marks, mark size=1pt, color=red] table [y=NOMA_2_1, x=I_dB,col sep = comma]{Data_Rate_vs_IdB_II.csv};
	\addlegendentry{I-I (2, 1)}
	\addplot[mark=o, only marks, mark size=2.5pt, color=red] table [y=NOMA_2_1, x=I_dB,col sep = comma]{Data_Rate_vs_IdB_IS.csv};
	\addlegendentry{I-S (2, 1)}
	\addplot[mark=*, only marks, mark size=1pt, color=gray] table [y=NOMA_2_1, x=I_dB,col sep = comma]{Data_Rate_vs_IdB_SI.csv};
	\addlegendentry{S-I (2, 1)}
	\addplot[mark=o, only marks, mark size=2.5pt, color=gray] table [y=NOMA_2_1, x=I_dB,col sep = comma]{Data_Rate_vs_IdB_SS.csv};
	\addlegendentry{S-S (2, 1)}
	\addplot[line width=0.5pt,solid,color=red] table [y=NOMA_2_1, x=I_dB,col sep = comma]{Data_Rate_vs_IdB_II.csv};
	\addplot[line width=0.5pt,solid,color=cyan] table [y=NOMA_2_4, x=I_dB,col sep = comma]{Data_Rate_vs_IdB_II.csv};
	\addplot[line width=0.5pt,solid,color=violet] table [y=NOMA_4_1, x=I_dB,col sep = comma]{Data_Rate_vs_IdB_II.csv};
	\addplot[line width=0.5pt,solid,color=orange] table [y=NOMA_4_4, x=I_dB,col sep = comma]{Data_Rate_vs_IdB_II.csv};
	\addplot[line width=0.5pt,solid,color=gray] table [y=NOMA_2_1, x=I_dB,col sep = comma]{Data_Rate_vs_IdB_SI.csv};
	\addplot[line width=0.5pt,solid,color=teal] table [y=NOMA_2_4, x=I_dB,col sep = comma]{Data_Rate_vs_IdB_SI.csv};
	\addplot[line width=0.5pt,solid,color=green] table [y=NOMA_4_1, x=I_dB,col sep = comma]{Data_Rate_vs_IdB_SI.csv};
	\addplot[line width=0.5pt,solid,color=magenta] table [y=NOMA_4_4, x=I_dB,col sep = comma]{Data_Rate_vs_IdB_SI.csv};
	\end{axis}
\end{tikzpicture}
}

%% file: Fig4.tex
\pgfplotsset{compat=1.9}
\resizebox{0.7\columnwidth}{!}{
\begin{tikzpicture}   
      \begin{axis}[
      	xmode=log,
	log basis x={10},
      	xlabel={Delay exponent $\theta$},
	xtick={0.01, 0.1, 1, 10, 100, 1000},
	xmin=0.01,xmax=1000,
	ylabel={Sum ER (bps/Hz)},
	ymin=0,ymax=11,
	grid=both,
	minor grid style={gray!25},
	major grid style={gray!25},
	legend columns=2, 
legend style={{nodes={scale=0.8, transform shape}}, at={(0.58,0.586)},  anchor=south west, draw=black,fill=white,legend cell align=left,inner sep=1pt,row sep = -2pt},]
	%%========== 4-users, 4-antennas==============
	\addplot[mark=triangle*, only marks, mark size=1.5pt, color=orange] table [y=NOMA_4_4, x=theta,col sep = comma]{Data_Rate_vs_theta_II.csv};
	\addlegendentry{I-I (4, 4)}
	\addplot[mark=triangle, only marks, mark size=3.5pt, color=orange] table [y=NOMA_4_4, x=theta,col sep = comma]{Data_Rate_vs_theta_IS.csv};
	\addlegendentry{I-S (4, 4)}
	\addplot[mark=triangle*, only marks, mark size=1.5pt, color=magenta] table [y=NOMA_4_4, x=theta,col sep = comma]{Data_Rate_vs_theta_SI.csv};
	\addlegendentry{S-I (4, 4)}
	\addplot[mark=triangle, only marks, mark size=3.5pt, color=magenta] table [y=NOMA_4_4, x=theta,col sep = comma]{Data_Rate_vs_theta_SS.csv};
	\addlegendentry{S-S (4, 4)}
	\addplot[mark=square*, only marks, mark size=1pt, color=violet] table [y=NOMA_4_1, x=theta,col sep = comma]{Data_Rate_vs_theta_II.csv};
	\addlegendentry{I-I (4, 1)}
	\addplot[mark=square, only marks, mark size=2.5pt, color=violet] table [y=NOMA_4_1, x=theta,col sep = comma]{Data_Rate_vs_theta_IS.csv};
	\addlegendentry{I-S (4, 1)}
	\addplot[mark=diamond*, only marks, mark size=1.5pt, color=cyan] table [y=NOMA_2_4, x=theta,col sep = comma]{Data_Rate_vs_theta_II.csv};
	\addlegendentry{I-I (2, 4)}
	\addplot[mark=diamond, only marks, mark size=3.5pt, color=cyan] table [y=NOMA_2_4, x=theta,col sep = comma]{Data_Rate_vs_theta_IS.csv};
	\addlegendentry{I-S (2, 4)}
	\addplot[mark=diamond*, only marks, mark size=1.5pt, color=teal] table [y=NOMA_2_4, x=theta,col sep = comma]{Data_Rate_vs_theta_SI.csv};
	\addlegendentry{S-I (2, 4)}
	\addplot[mark=diamond, only marks, mark size=3.5pt, color=teal] table [y=NOMA_2_4, x=theta,col sep = comma]{Data_Rate_vs_theta_SS.csv};
	\addlegendentry{S-S (2, 4)}
	\addplot[mark=square*, only marks, mark size=1pt, color=green] table [y=NOMA_4_1, x=theta,col sep = comma]{Data_Rate_vs_theta_SI.csv};
	\addlegendentry{S-I (4, 1)}
	\addplot[mark=square, only marks, mark size=2.5pt, color=green] table [y=NOMA_4_1, x=theta,col sep = comma]{Data_Rate_vs_theta_SS.csv};
	\addlegendentry{S-S (4, 1)}
	\addplot[mark=*, only marks, mark size=1pt, color=red] table [y=NOMA_2_1, x=theta,col sep = comma]{Data_Rate_vs_theta_II.csv};
	\addlegendentry{I-I (2, 1)}
	\addplot[mark=o, only marks, mark size=2.5pt, color=red] table [y=NOMA_2_1, x=theta,col sep = comma]{Data_Rate_vs_theta_IS.csv};
	\addlegendentry{I-S (2, 1)}
	\addplot[mark=*, only marks, mark size=1pt, color=gray] table [y=NOMA_2_1, x=theta,col sep = comma]{Data_Rate_vs_theta_SI.csv};
	\addlegendentry{S-I (2, 1)}
	\addplot[mark=o, only marks, mark size=2.5pt, color=gray] table [y=NOMA_2_1, x=theta,col sep = comma]{Data_Rate_vs_theta_SS.csv};
	\addlegendentry{S-S (2, 1)}
	\addplot[line width=0.5pt,solid,color=red] table [y=NOMA_2_1, x=theta,col sep = comma]{Data_Rate_vs_theta_II.csv};
	\addplot[line width=0.5pt,solid,color=cyan] table [y=NOMA_2_4, x=theta,col sep = comma]{Data_Rate_vs_theta_II.csv};
	\addplot[line width=0.5pt,solid,color=violet] table [y=NOMA_4_1, x=theta,col sep = comma]{Data_Rate_vs_theta_II.csv};
	\addplot[line width=0.5pt,solid,color=orange] table [y=NOMA_4_4, x=theta,col sep = comma]{Data_Rate_vs_theta_II.csv};
	\addplot[line width=0.5pt,solid,color=gray] table [y=NOMA_2_1, x=theta,col sep = comma]{Data_Rate_vs_theta_SI.csv};
	\addplot[line width=0.5pt,solid,color=teal] table [y=NOMA_2_4, x=theta,col sep = comma]{Data_Rate_vs_theta_SI.csv};
	\addplot[line width=0.5pt,solid,color=green] table [y=NOMA_4_1, x=theta,col sep = comma]{Data_Rate_vs_theta_SI.csv};
	\addplot[line width=0.5pt,solid,color=magenta] table [y=NOMA_4_4, x=theta,col sep = comma]{Data_Rate_vs_theta_SI.csv};
	\end{axis}
\end{tikzpicture}
}